\documentclass[runningheads]{llncs}
\usepackage{graphicx,comment,inputenc,amsmath,mathtools,xcolor,amssymb,tabularx,siunitx,textcomp,float,multirow,accents,bm,bbm,algorithm,algpseudocode,caption,enumitem,caption,subcaption,tikz,pgfplots,graphics}
\usepackage{hyperref}
\usepackage[thinc]{esdiff}
\usetikzlibrary{shapes, arrows, scopes, chains, pgfplots.groupplots}
\usepgfplotslibrary{fillbetween}

\DeclareMathOperator*{\argmax}{arg\,max}

\newcolumntype{L}[1]{>{\raggedright\arraybackslash}p{#1}}
\newcolumntype{C}[1]{>{\centering\arraybackslash}p{#1}}
\newcolumntype{R}[1]{>{\raggedleft\arraybackslash}p{#1}}
\newcommand\numberthis{\addtocounter{equation}{1}\tag{\theequation}}

\makeatletter
\newcommand{\pushright}[1]{\ifmeasuring@#1\else\omit\hfill$\displaystyle#1$\fi\ignorespaces}
\newcommand{\pushleft}[1]{\ifmeasuring@#1\else\omit$\displaystyle#1$\hfill\fi\ignorespaces}
\makeatother
\begin{document}
\bibliographystyle{splncs04}
\title{Regret-Optimal Online Caching for Adversarial and Stochastic Arrivals\thanks{This work is supported by a SERB grant on Leveraging Edge Resources for Service Hosting.}}
\titlerunning{Optimal switching regret for online caching}
%
\author{Fathima Zarin Faizal\inst{1}\orcidID{0000-0002-5663-8308} \and
Priya Singh\inst{1}\orcidID{0000-0002-1658-1116} \and
Nikhil Karamchandani \inst{1}\orcidID{0000-0002-7233-0717} \and
Sharayu Moharir \inst{1}\orcidID{0000-0001-9393-9276} }
\authorrunning{F. Faizal et al.}
%
\institute{Indian Institute of Technology Bombay, Mumbai, Maharashtra, India - 400076\\
\url{iitb.ac.in} }
\maketitle              
\begin{abstract}
We consider the online caching problem for a cache of limited size. In a time-slotted system, a user requests one file from a large catalog in each slot. If the requested file is cached, the policy receives a unit reward and zero rewards otherwise. We show that a Follow the Perturbed Leader (FTPL)-based anytime caching policy is simultaneously regret-optimal for both adversarial and i.i.d. stochastic arrivals. Further, in the setting where there is a cost associated with switching the cached contents, we propose a variant of FTPL that is order-optimal with respect to time for both adversarial and stochastic arrivals and has a significantly better performance compared to FTPL with respect to the switching cost for stochastic arrivals. We also show that these results can be generalized to the setting where there are constraints on the frequency with which cache contents can be changed. Finally, we validate the results obtained on various synthetic as well as real-world traces. \keywords{Online caching \and algorithms \and regret bounds}
\end{abstract}
\section{Introduction}
The caching problem has been studied since the 1960s, initially motivated by memory management in computers \cite{computer_stuff}. More recently, there has been renewed interest motivated by Content Delivery Networks \cite{cdn} used for applications such as Video-on-Demand services. Such applications rely on low latency to provide a good customer experience. The framework of this problem involves a library of $L$ files and a cache located near the end-users that is capable of storing at most $C$ files at any given time, the algorithmic challenge being to determine the most popular files to be stored in the cache.

Two types of arrival patterns have been considered in the existing literature and in our work. The first is known as the Independent Reference Model, where request arrivals are generated by an i.i.d. stochastic process and the distribution of the request process is unknown to the policy. The second arrival model is the adversarial arrival model where we make no structural assumptions on the arrival process. Here, the arrival pattern is generated by an oblivious adversary who knows which caching policy is being used but is not aware of the sample path of decisions made by the policy. In both models, as we are focused on the online caching problem, requests are revealed causally and therefore caching decisions have to be made based on past arrival patterns without any explicit knowledge of future arrivals.

Various metrics have been used to characterize the performance of caching polices, including hit-rate and competitive ratio. Regret is a popular metric for online learning algorithms \cite{cesa_lugosi} and is defined as the difference between the reward incurred by the optimal stationary policy and the policy under consideration. Our broad goal is to determine if there exists caching policies that have order-optimal regret with respect to time for i.i.d. stochastic and adversarial arrivals and therefore robust to the nature of the arrival process. Polices that perform well in the adversarial setting are primarily focused on not performing horribly on any arrival sequence. This often leads to sub-optimal performance for specific arrival sequences. Similarly, policies designed for specific arrival processes or under some structural assumptions on the arrival process have poor performance in the adversarial setting as they have very poor performance for specific arrival processes which affects the worst-case performance of the policy. For instance, policies designed for the independent reference model would not be ideal when requests are not stationary. 

Prediction with expert advice \cite{cesa_lugosi,littlestone1994} and Online Convex Optimization \cite{zinkevich2003online} are well-known settings in online learning for which optimal algorithms have been found. Though the caching problem is equivalent to the prediction with expert advice setting with ${L \choose C}$ experts, ${L \choose C}$ is typically a very large number resulting in standard algorithms being computationally inefficient. Least Frequently Used (LFU), Least Recently Used (LRU) and First-in-First-Out (FIFO) are popular caching policies that have been shown to achieve optimal competitive ratio \cite{albers}. There are also results on the closed form stationary hit probabilities of these algorithms under the Independent Reference Model \cite{stoch_analysis_comp_storage,lru_climb_stat_hit_prob}. 

Under stochastic arrivals, LFU achieves order-optimal regret \cite{learning_to_cache} but under adversarial arrivals, LFU, LRU and FIFO have been shown to have suboptimal regret \cite{lfu_lb_paper}. A sublinear regret upper bound was proved for a gradient-based coded caching policy (OGA) under adversarial arrivals \cite{lfu_lb_paper} while the first uncoded caching policy to be shown to achieve sublinear regret is the Follow The Perturbed Leader (FTPL) policy \cite{sigmetrics}. Proposed in \cite{hannan1957}, FTPL has also been shown to achieve order-optimal regret under adversarial arrivals by proving a lower bound on the regret using a balls-into-bins argument in \cite{sigmetrics}. An FTPL-based policy has also been shown to be regret-optimal for bipartite caching networks \cite{leadcache}. 

These policies do not consider the overhead of fetching files into the cache from the library each time the cache updates, called the \emph{switching cost} \cite{sqrt_t}. An $\tilde{O}(C \sqrt{T})$ upper bound on the regret including the switching cost was shown for a variant of the Multiplicative-Weight policy (MW) under adversarial arrivals which is also more computationally efficient compared to the original MW algorithm naively applied to the caching problem. An upper bound of $\tilde{O}(\sqrt{CT})$ was shown for an FTPL-based policy which is also simpler to implement \cite{sqrt_t} and improves upon the earlier bound by a factor of $\Theta(\sqrt{C})$. 

\subsection{Our contributions}
 We consider the following two settings: unrestricted switching, where the objective is to minimize the regret including the switching cost, and restricted switching, where the cache is allowed to update at certain fixed points only and the objective is to minimize regret. In Section \ref{sec: unlimited}, we consider the unrestricted switching setting and show that FTPL with an adaptive learning rate achieves order-optimal regret under stochastic arrivals even after including the switching cost, while FTPL with a constant learning rate cannot have order-optimal regret for both stochastic and adversarial arrivals. We also propose the Wait then FTPL (W-FTPL) policy that improves the bound on the switching cost from $\mathcal{O}(D)$ to $\mathcal{O}(\log D)$, where $D$ is the per-file switching cost. In Section \ref{sec: restricted}, we consider the restricted switching setting and prove a lower bound on the regret of any policy and an upper bound on the regret of FTPL. We show that FTPL acheives order-optimal regret under stochastic file requests and in a special case of this setting under adversarial file requests. Finally, in Section \ref{sec: numerical}, we present the results of numerical experiments on synthetic as well as real-world traces that validate the results obtained. Due to a lack of space, the proofs of the theorems stated in this paper can be found in \cite{full_paper}. 

We thus show that FTPL with an adaptive learning rate applied to the online caching problem has order-optimal regret under stochastic and adversarial arrivals in the unrestricted switching and in a special case of the restricted switching setting. 
\section{Problem formulation}
We consider the classical content caching problem where a user requests files from a library that is stored in a back-end server. There is a cache that is capable of serving user requests at a lower cost but has a storage size that is typically considerably smaller than the library size. Time is slotted and in each time slot, the user requests at most one file. The sequence of events in a time slot is as follows. The cache may first update its contents, after which it receives a request for a file from the user. If the requested file is available in the cache, the cache is said to have a \emph{hit} and the request is fulfilled locally by the cache, and otherwise a \emph{miss}, in which case the file request is fulfilled by the back-end server. 

\emph{Cache configuration.} We consider a cache of size $C$ that stores files from a library $\mathcal{L}$ of size $L$. Usually, the cache size is much smaller than the library size, i.e., $C \ll L$. The file requested by the user at time $t$ is denoted by $x_t$ and is represented also in the form of the one-hot encoded vector $\mathbf{x}_t \in \{0,1\}^{L}$. For $\tau \geq 2$, we denote by $\boldsymbol{X}_{\tau}=\sum_{t=1}^{\tau-1} \boldsymbol{x}_t$ the $L$-length vector storing the cumulative sum of requests for each file till time slot $\tau-1$. $\boldsymbol{X}_1$ is initialized to be the zero vector. Let C(t) denote the set of files cached in round $t$ and let $\boldsymbol{y}_t \in \{0,1\}^{L}$ be a binary vector denoting the state of the cache at time $t$, such that $y_t = (y_t^1, y_t^2, \ldots, y_t^L)$ with $y_t^i= 1$ for $i \in C(t)$ and 0 otherwise. 

\emph{File requests.} We consider two types of file requests: adversarial and stochastic. The file requests are said to be \textit{adversarial} if no assumptions are made regarding the statistical properties of the file requests. We assume that the adversary is oblivious, i.e., the entire file request sequence is fixed before the first request is sent. The file requests are said to be \textit{stochastic} if in each slot, the request is generated independently according to a popularity distribution $\boldsymbol{\mu}=(\mu_1, \ldots, \mu_L)$, where $\mathbb{P}(x_t=i)=\mu_i$ and $\sum_{i} \mu_i =1$. Without loss of generality, we assume that $\mu_1 \geq \ldots \geq \mu_L$. As is the case in most real-world applications, the popularity distribution is assumed to be unknown to the caching policy beforehand. 

\emph{Caching Policy.} At the beginning of any given time slot $t \geq 2$, a caching policy $\pi(\cdot)$ maps the history of observations it has seen so far (denoted by $h(t)$) to a valid cache configuration $C(t)$, i.e., $C(t)=\pi(h(t))$. In the first time slot, we assume that the cache stores $C$ files randomly chosen from the library and that this does not incur any switch cost. We define $T$ to be the time horizon of interest. When the file requests are adversarial, the optimal stationary policy is defined to be the caching policy that stores the top $C$ files in hindsight, i.e., stores the $C$ files that received the maximum number of requests till time $T$. When the file requests are stochastic, the optimal stationary policy is defined to be the caching policy that stores the files with the top $C$ popularities in the cache, i.e., $C(t)=\mathcal{C} \, \forall t$, where $\mathcal{C}=\{1,\ldots,C\}$. 

\emph{Reward and Switch Cost.} At each time step, the policy obtains a reward of 1 unit when the requested file is available in the cache, i.e., a hit occurs, and a reward of 0 units otherwise. A caching policy that fetches a large number of files into the cache each time the cache updates is not ideal as fetching files into the cache causes latency and consumes bandwidth. Thus, we also consider the switch cost, i.e., the cost of fetching files from the back-end server into the cache. We assume that fetching a file into the cache from the server incurs a cost of $D$ units. 

\emph{Performance metric.} Policies are evaluated on the basis of the regret that they incur. Informally, the regret of a policy till time $T$ is the difference between the net utility of the optimal stationary policy and the net utility the policy under consideration. The net utility of a policy till time $T$ is the difference between the net reward accrued till $T$ and the overall switch cost incurred till then. Stationary policies do not incur any switch cost and hence their net utility is determined by the overall number of hits they have till time $T$. As discussed in the next section, in some cases, we omit the switch cost. 

\emph{Problem settings.} We consider the following two variations of the classical content caching problem:
\begin{enumerate}  
\item \textbf{Setting 1:} Unrestricted switching with switching cost. \\
In this setting, the system incurs a cost of $D$ units every time a file is fetched from the back-end server to be stored in the cache. When following a policy $\pi$ on a request sequence $\{x_t\}_{t=1}^{T}$, the regret till time $T$ for $\{x_t\}_{t=1}^{T}$ including the switching cost when the file requests are adversarial is defined as:
\begin{align}
 R^{\pi}_{A}(\{x_t\}_{t=1}^{T},T,D) &=\sup _{\boldsymbol{y} \in \mathcal{Y}}\left\langle\boldsymbol{y}, \boldsymbol{X}_{T+1}\right\rangle-\sum_{t=1}^{T} \mathbb{E}\left\langle\boldsymbol{y}_{t}, \boldsymbol{x}_{t}\right\rangle +\frac{D}{2} \sum_{t=1}^{T-1} \mathbb{E}\left\|\boldsymbol{y}_{t+1}-\boldsymbol{y}_{t}\right\|_{1},
\end{align}
where the expectation is with respect to any randomness introduced by the policy. The regret of a policy $\pi$ till time $T$ is defined as the worst-case regret over all possible request sequences, i.e.,
\begin{align*}
 R^{\pi}_{A}(T,D) = \underset{\{x_t\}_{t=1}^{T}}{\sup}  R^{\pi}_{A}(\{x_t\}_{t=1}^{T},T,D).
\end{align*}
When the file requests are stochastic, the regret including the switching cost after $T$ time steps is defined as:
\begin{align}
    R^{\pi}_{S}(T,D) =\mathbb{E} \left [ \sum_{t=1}^{T} \mathbbm{1}\{x(t) \in \mathcal{C}\}-\mathbbm{1}\{x(t) \in C(t)\} \right ] +\frac{D}{2} \sum_{t=1}^{T-1} \mathbb{E}\left\|\boldsymbol{y}_{t+1}-\boldsymbol{y}_{t}\right\|_{1}, \label{eqn: switching_regret_stochastic}
\end{align}
where $\mathcal{C}$ denotes the set of files having the top $C$ popularities. In the above expression, the expectation is taken with respect to the randomness in the file requests as well as any randomness introduced by the policy.
\newline
\item \textbf{Setting 2:} Restricted switching without switching cost. \\
Here, the cache is allowed to change its contents only at $s+1$ fixed time slots for some $1 \leq s \leq T$. To be precise, the cache is allowed to change its contents only at the beginning of the following time slots: $1, r_1+1, r_1+r_2+1, \ldots, \sum_{i=1}^{s}r_i+1$, where $1 \leq r_i \leq T, 1 \leq i \leq s$ denotes the $i^{\text{th}}$ inter-switching period such that $\sum_{i=1}^{s}r_i = T$. Thus, within the time horizon $T$, the cache is allowed to update only at $s$ fixed time slots. Note that the setting where the cache is allowed to change its contents only after every $1 \leq r \leq T$ requests, i.e., at time slots $1,r+1,\ldots,T+1$ and $s=\frac{T}{r}$ is a special case of this setting. For simplicity, we restrict our attention to the case where there is no switch cost, i.e., $D=0$. When following a policy $\pi$, the regret after $T$ time steps when the file requests are adversarial is:
\begin{align}
 R^{\pi}_{A}(T) &=\sup _{\boldsymbol{y} \in \mathcal{Y}}\left\langle\boldsymbol{y}, \boldsymbol{X}_{T+1}\right\rangle-\sum_{t=1}^{T} \mathbb{E}\left\langle\boldsymbol{y}_{t}, \boldsymbol{x}_{t}\right\rangle,
\end{align}
where the expectation is with respect to the randomness introduced by the policy. When the file requests are stochastic, the regret after $T$ time steps is:
\begin{align}
    R^{\pi}_{S}(T) =\mathbb{E} \left [ \sum_{t=1}^{T} \mathbbm{1}\{x(t) \in \mathcal{C}\}-\mathbbm{1}\{x(t) \in C(t)\} \right ], \label{eqn: switching_regret_stochastic}
\end{align}
where $\mathcal{C}$ denotes the set of files having the top $C$ popularities. In the above expression, the expectation is taken with respect to the randomness introduced by the policy and the file requests.
\end{enumerate} 
To distinguish between results including switching cost and those without switching cost, we use the notation $R_{(\cdot)}^{\pi}(T,D)$ for results involving the switch cost and $R_{(\cdot)}^{\pi}(T)$ for results without the switch cost. 

The overall goal of this work is to characterize the optimal regret in the two settings mentioned above, for both adversarial and stochastic file requests. This entails proving scheme-agnostic lower bounds on the regret as well as designing policies whose regret is of the same order as these lower bounds. As we will see, these results will also highlight the impact of switching cost and intermittent switching on the optimal achievable regret. 
\section{Policies} \label{sec: policies}
In this section, we introduce and formalize policies whose optimality (or suboptimality) will be discussed in later sections. 
\subsection{Least Frequently Used (LFU)}
The LFU algorithm (formally defined in Algorithm \ref{alg: lfu}) keeps track of the number of times each file has been requested so far. At each time step $t$, the files with the $C$ highest number of requests are cached. This policy is deterministic and thus performs poorly when faced with certain adversarial request sequences \cite{lfu_lb_paper}. For the simplified case of $L=2, C=1$, one example is a round-robin request sequence of the form $1,2,1,2,\ldots$ which would result in LFU obtaining essentially zero reward while the optimal stationary policy obtains a reward of $T/2$. For stochastic requests, it has been shown to achieve $\mathcal{O}(1)$ regret when switching is allowed at all time slots and when the algorithm incurs no switch cost \cite{learning_to_cache}. 
\begin{algorithm}[H]
\caption{LFU algorithm}\label{alg: lfu}
\begin{algorithmic}[1]
\Procedure{LFU}{$T$}
\State $\boldsymbol{c}_t \gets \mathbf{0}$
\While{$t \leq T$}
\State $C_t \gets \underset{C}{\argmax} (c_t(1), \ldots, c_t(L)  )$
\State Receive file request $x_t$
\State $ c_t(x_t) \gets  c_t(x_t) +  1$
\EndWhile
\EndProcedure
\end{algorithmic}
\end{algorithm}
\subsection{Follow The Perturbed Leader (FTPL)} 
The FTPL algorithm (formally defined in Algorithm \ref{alg: ftpl})) is a variation of the LFU algorithm and also keeps track of the number of times each file has been requested so far, but adds an independent Gaussian perturbation with mean 0 and standard deviation $\eta_t$ (referred to as the learning rate) to the counts of each file in each time slot. At each time step $t$, the files with the $C$ highest perturbed counts are cached. Special cases of this policy are known to achieve order-optimal regret under adversarial requests (with and without switch cost) \cite{sigmetrics,sqrt_t}. We will henceforth refer to the FTPL algorithm with the learning rate $\eta_t$ by FTPL($\eta_t$). 
\begin{algorithm}[H]
\caption{FTPL algorithm}\label{alg: ftpl}
\begin{algorithmic}[1]
\Procedure{FTPL}{$T,\{\eta_{t}\}_{t=1}^{T}$}
\State $\boldsymbol{c}_t \gets \mathbf{0}$
\State Sample $\boldsymbol{\gamma} \sim \mathcal{N}(\boldsymbol{0}, \boldsymbol{I}_{L \times L})$
\While{$t \leq T$}
\State $C_t \gets \underset{C}{\argmax} (c_t(1) + \eta_t\gamma(1) , \ldots, c_t(L) + \eta_t\gamma(L) )$
\State Receive file request $x_t$
\State $ c_t(x_t) \gets  c_t(x_t) +  1$
\EndWhile
\EndProcedure
\end{algorithmic}
\end{algorithm}
\subsection{Wait then FTPL (W-FTPL)} 
The algorithm that we propose, Wait then FTPL (formally defined in Algorithm \ref{alg: wftpl})), is a variant of the FTPL algorithm where the policy remains idle for an initial deterministic waiting period and then follows the normal FTPL algorithm. The motivation for this algorithm is to avoid the higher switch cost incurred initially by the FTPL algorithm under stochastic file requests until the policy has seen enough requests to have a good enough estimate of the underlying popularity distribution, while ensuring order-optimal regret in the adversarial setting. We will henceforth refer to the W-FTPL algorithm with the learning rate $\eta_t$ by W-FTPL($\eta_t$). 
\begin{algorithm}[H]
\caption{W-FTPL algorithm}\label{alg: wftpl}
\begin{algorithmic}[1]
\Procedure{W-FTPL}{$T,\{\eta_{t}\}_{t=1}^{T},D,t'$}
\State $\boldsymbol{c}_t \gets \mathbf{0}$
\State Sample $\boldsymbol{\gamma} \sim \mathcal{N}(\boldsymbol{0}, \boldsymbol{I}_{L \times L})$
\While{$t \leq T$}
\If{$t >t'$}
\State $C_t \gets \underset{C}{\argmax} (c_t(1) + \eta_t\gamma(1) , \ldots, c_t(L) + \eta_t\gamma(L) )$
\EndIf
\State Receive file request $x_t$
\State $ c_t(x_t) \gets  c_t(x_t) +  1$
\EndWhile
\EndProcedure
\end{algorithmic}
\end{algorithm}
\section{Setting 1: Unrestricted switching with switching cost} \label{sec: unlimited}
In this section, we consider the setting where there is no limitation on the switching frequency of the cache and the objective is to minimize the regret including the switching cost, i.e., minimize regret as well as the number of fetches into the cache. We consider both stochastic and adversarial file request sequences and show that FTPL($\alpha \sqrt{t}$) and W-FTPL($\alpha \sqrt{t}$) are order-optimal under both types of file requests. While the FTPL($\eta$) algorithm is order-optimal under adversarial requests for a particular value of $\eta$ \cite{sigmetrics}, we prove that the same does not hold true for stochastic file requests. 
\subsection{Adversarial requests} \label{subsec: unlimited_adversarial}
In this section, we discuss the performance of the policies introduced in Section \ref{sec: policies} under adversarial file requests. The key results of this section has been summarized in the following theorem:
\begin{theorem} \label{thm: unlimited_adversarial}
Under adversarial requests, we have
\begin{enumerate}[label=(\alph*)]
\item \cite[Theorem~2]{sigmetrics} For any policy $\pi$ and $L \geq 2C$,
\begin{align*}
R_{A}^{\pi}(T,D=0) \geq \sqrt{\frac{C T}{2 \pi}}-\Theta\left(\frac{1}{\sqrt{T}}\right). 
\end{align*}
\item \cite[Proposition~1]{lfu_lb_paper} The regret of the LFU policy can be characterized as:
 \begin{align*}
R_{A}^{\textrm{LFU}}(T,0) = \Omega(T).
\end{align*}
\item \cite[Theorem~4.1]{sqrt_t} The regret of FTPL($\alpha \sqrt{t}$) is upper bounded as:
\begin{align*}
R_{A}^{\textrm{FTPL}(\alpha \sqrt{t})}(T,D)  \leq c_{1} \sqrt{T}+c_{2} \ln T+c_{3},
\end{align*}
where $c_{1}=\mathcal{O}(\sqrt{\ln (L e / C)})$, and $c_{2}, c_{3}$ are small constants depending on $L, C,D$ and $\alpha$.  
\item The regret of W-FTPL($\alpha \sqrt{t}$) is upper bounded as:
 \begin{align*}
R_{A}^{\textrm{W-FTPL}(\alpha \sqrt{t})}(T,D) \leq  \mathcal{O}(\sqrt{T}).
\end{align*}
\end{enumerate}
\end{theorem}
Part (a) has been proved in \cite{sigmetrics} and provides a lower bound on the regret of any policy under adversarial requests. 

As argued before, LFU performs poorly under adversarial requests. This is also seen for many popular classical caching algorithms like LRU and FIFO (refer \cite{lfu_lb_paper}).

Part (c) has been proved in \cite{sqrt_t} and provides an $\mathcal{O}(\sqrt{T})$ upper bound on the regret including the switching cost of the FTPL($\alpha \sqrt{t}$) policy under adversarial requests, thus showing that this algorithm is order-optimal under adversarial requests. FTPL($\alpha \sqrt{T}$) has also been shown to be order-optimal under adversarial requests \cite{sigmetrics,sqrt_t}. 

Part (d) provides an upper bound on the regret including the switching cost of W-FTPL($\alpha \sqrt{t}$) under adversarial requests. This result shows that W-FTPL($\alpha \sqrt{t}$) is order-optimal under adversarial requests. The proof of this result can be found in Appendix \ref{sec: app_wftpl_adversarial_ub}.

\subsection{Stochastic requests}
To find a policy that achieves order-optimal regret under stochastic and adversarial arrivals, we were motivated by \cite{sigmetrics,sqrt_t} where the regret for FTPL with the learning rates $\alpha \sqrt{t}$ and $\alpha \sqrt{T}$ ($\alpha$ being some positive constant) under adversarial arrivals was characterized. In this section, we discuss the performance of these policies under stochastic file requests.  The key results of this section has been summarized in the following theorem:
\begin{theorem} \label{thm: unlimited_stochastic}
The file requests are stochastic i.i.d. with the popularity distribution $\boldsymbol{\mu}$. 
\begin{enumerate}[label=(\alph*)]
\item \cite[Theorem~1]{learning_to_cache} When $D=0$, the regret of the LFU policy can be upper bounded as:
\begin{align}
R^{LFU}_{S}(T,0)<\min \left(\frac{16}{\Delta_{\min }^{2}}, \frac{4 C(L-C)}{\Delta_{\min }}\right), \label{eqn: lfu_ub}
\end{align}
where $\Delta_{\min}=\mu_C-\mu_{C+1}$. 
\item For $L=2, C=1$ and $D=0$, the regret of FTPL($\eta$) can be lower bounded as:
\begin{align*}
     R^{\textrm{FTPL}(\eta)}_{S}(T,0) \geq \frac{\eta e^{-\left ( \frac{1+\eta}{\eta} \right )^2}}{4 }.
\end{align*}
\item The regret of FTPL($\alpha \sqrt{t}$) is upper bounded as:
 \begin{align*}
    R^{\textrm{FTPL}(\sqrt{t})}_{S}(T,D) \leq \left (1+ DC \right ) t_0 + \left (1+ \frac{D}{\Delta_{\min}} \right ) \left ( \frac{8}{\Delta_{\min}} + \frac{32 \alpha^2}{\Delta_{\min}} \right ),
\end{align*}
where $t_0=\max \left \{\frac{8}{\Delta_{\min}^2} \log \left ( {L^3}\right )  , \frac{32 \alpha^2}{\Delta_{\min}^2 }\log \left ( {L^3}\right ) \right \}$. 
\item The regret of W-FTPL($\alpha \sqrt{t}$) is upper bounded as:
 \begin{align*}
    R^{\textrm{W-FTPL}(\sqrt{t})}_{S}(T,D) &\leq t' + \frac{16}{\Delta_{\min }} + \frac{64 \alpha^2}{\Delta_{\min }}+ 2L^3D \biggl ( e^{- u (\log D)^{1+\beta} \Delta_{\min}^{2} / 8} \frac{8}{\Delta_{\min}^2} \\&+ e^{-u (\log D)^{1+\beta} \Delta_{\min}^{2} / 32 \alpha^{2}} \frac{32 \alpha^2}{\Delta_{\min}^2} \biggr ),
\end{align*}
where $t'=\max \left \{\frac{8}{\Delta_{\min}^2} \log \left ( \frac{L^3}{2}\right )  , \frac{32 \alpha^2}{\Delta_{\min}^2 }\log \left ( \frac{L^3}{2}\right ), u (\log D)^{1+\beta} \right \}$. 
\end{enumerate}
\end{theorem}

Part (a) of the above theorem has been proved in \cite{learning_to_cache} and shows that the regret of LFU is $\mathcal{O}(1)$ when the file requests are stochastic. Thus, the regret of any policy that has order optimal regret under stochastic file requests should be $\mathcal{O}(1)$.  

Part (b) gives a lower bound on the regret of the FTPL($\eta$) algorithm under stochastic file requests. Note that for all $\eta \geq 1$, we have
\begin{align*}
    R^{\textrm{FTPL}(\eta)}_{S}(T)&\geq \frac{\eta e^{-\left ( \frac{1+\eta}{\eta} \right )^2}}{4} \geq \frac{\eta}{4 e^4}. 
\end{align*}
This shows that the regret is $\Omega(\eta)$ for the FTPL($\eta$) algorithm when the file requests are stochastic. Using $\eta=\alpha \sqrt{T}$, where $\alpha$ is a positive constant, from \cite{sigmetrics} which gave an $\mathcal{O}(\sqrt{T})$ upper bound for FTPL($\eta$), we get that with the constant learning rate $\eta=\alpha \sqrt{T}$, the regret of FTPL($\eta$) is $\Theta(\sqrt{T})$. Thus, while FTPL($\eta$) is order-optimal under adversarial requests (see Section \ref{subsec: unlimited_adversarial}), it cannot simultaneously be order-optimal under adversarial and stochastic file requests. The proof of this result can be found in Appendix \ref{sec: app_ftpl_lb}.

Part (c) shows that FTPL($\alpha \sqrt{t}$), $\alpha >0$ is order-optimal when the file requests are stochastic. While FTPL($\eta$) with $\eta=\mathcal{O}(\sqrt{T})$ achieves $\Omega(\sqrt{T})$ regret, this result shows an $\mathcal{O}(1)$ upper bound on the regret of FTPL($\alpha \sqrt{t}$) including the switching cost, thus showing that this algorithm is order-optimal. Note that the upper bound grows linearly with the per-file switch cost $D$. The proof of this result can be found in Appendix \ref{sec: app_ftpl_ub_stochastic}

Recall that the Wait then FTPL($\alpha \sqrt{t}$) (W-FTPL($\alpha \sqrt{t}$)) algorithm is a variant of the FTPL($\alpha \sqrt{t}$) algorithm. The algorithm remains idle till time $t'=u (\log D)^{1+\beta}, u>0$, and then normal FTPL($\alpha \sqrt{t}$) is followed. Part (d) proves an $\mathcal{O}(1)$ upper bound on the regret including the switching cost of this algorithm with respect to the horizon $T$ under stochastic file requests, thus showing that this algorithm is order-optimal. The main improvement over the FTPL($\alpha \sqrt{t}$) algorithm is the $\mathcal{O}( (\log D)^{1+\beta} )$ upper bound on the regret including the switching cost for a large enough value of $D$ under stochastic file requests, as compared to the upper bound $\mathcal{O}(D)$ for FTPL($\alpha \sqrt{t}$). The key idea behind remaining idle for an initial period that depends logarithmically on $D$ is to avoid the higher switch cost incurred at the beginning by the FTPL($\alpha \sqrt{t}$) algorithm (refer to Section \ref{sec: numerical}). The proof of this result can be found in Appendix~\ref{sec: app_wftpl_ub}.
\section{Restricted switching} \label{sec: restricted}
In this section, we consider the setting where the cache is allowed to update its contents only at $s+1$ fixed number of time slots, where $s \in \mathbb{Z}, s \leq T$. The first point is at time slot $1$, the second at time slot $r_1+1$, the third point is at time slot $r_1+r_2+1$, and so on till the ${s+1}^{\mathrm{th}}$ point, which is at time slot $\sum_{i=1}^{s}r_i+1$, where $r_i \in \mathbb{Z}, r_i \leq T, 1 \leq i \leq s$ and $\sum_{i=1}^{s}r_i = T$. Note that the cache is allowed to update its contents only $s$ times till the time horizon $T$. Refer to Figure~\ref{fig: restricted} for an illustration of this setting. As a special case of this setting, we also consider the homogenous case where the cache is allowed to update only after every $r \in \mathbb{Z}$ requests, i.e., $s=\frac{T}{r}$. We study the regret performance of FTPL and also provide lower bounds on the regret incurred by any online scheme. In the homogenous case, we also show that FTPL($\sqrt{rt}$) achieves order-optimal regret. 
\begin{figure}[t]
\centering
\scalebox{0.74}{
\begin{tikzpicture}[ every text node part/.style={align=center}]
                    
  \begin{scope}[start chain=1, node distance=9mm,  every node/.style = draw]
  { 
    \begin{scope}[every join/.style=-, node distance=4mm]
    {
    \node [ellipse, on chain,fill=red!10] {Backend \\ Server};
    \node [ellipse, on chain, join,fill=blue!10] {Cache};
    }
    \end{scope}       
    
    \node [name=arrow, single arrow, on chain=1, single arrow head extend=3pt, minimum height=7mm, rotate = 180] {};
    \node [name=1, draw,on chain=1, minimum size=.9cm,fill=yellow!40] {};

    \begin{scope}[node distance=0mm, minimum size=.9cm]
    { 
        \foreach \x in {2,3,4} {
          \node [name = \x, draw,on chain=1] {};
        }
        \node [name=5, draw,on chain=1,fill=yellow!40] {};
                \foreach \x in {6,7} {
          \node [name = \x, draw,on chain=1] {};
        }
          \node [name=8, draw,on chain=1,fill=yellow!40] {};
                \foreach \x in {9,10} {
          \node [name = \x, draw,on chain=1] {};
        }
    }
    \end{scope}

    \node [name=k, single arrow, draw, on chain=1, single     arrow head extend=3pt, minimum height=7mm, rotate =     180]  {};
    \node [circle, on chain=1,fill=green!10] {User};
    
   }  
   \end{scope}
     
   \begin{scope}[start chain=2 going right, node distance=0cm, minimum size=0.9cm]
   {    
      \node[xshift = 4cm, yshift=0.26cm] at (1.north) {Request sequence};
        
      \node[xshift = -0.6cm, yshift=-0.26cm] at (1.south) {\small $t\ =$};
        \node[yshift=-0.3cm, on chain=2] at (1.south) {\tiny 1};
        \node[on chain=2] {\tiny 2};
        \node[on chain=2] {\hspace{0.05cm}\ldots};
        \node[on chain=2] {\tiny $r_1$};
         \node[on chain=2] {\tiny $r_1+1$};
        \node[on chain=2] {\ldots};
        \node[on chain=2] {\hspace{-0.05cm}\tiny $r_1+r_2$};
         \node[on chain=2,yshift=-0.18cm] {\hspace{-0.2cm} \tiny $r_1+r_2$\\ \tiny$+1$};
        \node[on chain=2,yshift=0.18cm] {\hspace{-0.15cm}\ldots};
        \node[on chain=2] {\hspace{-0.2cm}\tiny $T$};
   }
   \end{scope}
    
\end{tikzpicture}
}
\caption{The time slots where the cache is allowed to update its contents have been marked in yellow.}
\label{fig: restricted}
\end{figure}
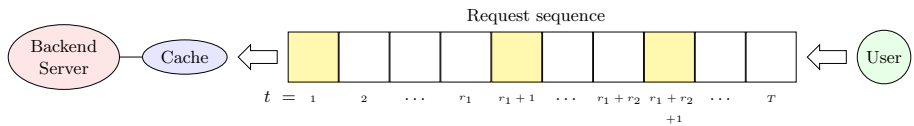
\subsection{Stochastic requests}
\begin{theorem}  \label{thm: restricted_stochastic_main}
The file requests are stochastic i.i.d. with the popularity distribution $\boldsymbol{\mu}$. When cache updates are restricted to $s+1$ fixed points defined by the inter-switching periods $\{r_i\}_{i=1}^{s}$ as outlined above, 
\begin{enumerate}[label=(\alph*)]
\item When $L=2, C=1$, for any online caching policy $\pi$, there exists a popularity distribution such that the popularities of the two files are greater than $1>a>0$ and the difference in the popularities is $\Delta$, such that
\begin{align*}
R^{\pi}_S(T) \geq \frac{r_1 \Delta}{2} +  \sum_{i=2}^{s} r_i \frac{\Delta}{4} \exp \left(- t_i \frac{\Delta^2}{a^2} \right).
\end{align*}
\item The regret of FTPL($\alpha \sqrt{t}$) is upper bounded as: 
\begin{align*}
R_{S}^{FTPL(\alpha \sqrt{t})}(T) &\leq r_1+ 2\sum_{j=1}^{C} \sum_{k=C+1}^{L} \sum_{i=2}^{s} r_i \, \Delta_{j, k} \left ( e^{-t_i \Delta_{j, k}^{2} / 8} + e^{-t_i \Delta_{j,k}^2/32 \alpha^2} \right ).
\end{align*}
\end{enumerate}
\end{theorem}
In part (a), we prove a fundamental lower bound on the regret of any policy $\pi$ under stochastic file requests when cache updates are restricted to $s+1$ fixed time slots. The proof of this result can be found in Appendix \ref{sec: app_restricted_stochastic_lb}. In part (b), we prove an upper bound on the regret of the FTPL($\alpha \sqrt{t}$) policy when cache updates are restricted to $s+1$ fixed time slots. The proof of this result can be found in Appendix \ref{sec: app_restricted_stochastic_ub_ftpl}. We thus have that the FTPL($\alpha \sqrt{t}$) policy has order-optimal regret in this setting under stochastic file requests. Next, we consider the special case where all $r_i$ are equal to $r$, i.e., $s=T/r$. 
\begin{theorem} \label{thm: constrained_stochastic_main}
The file requests are stochastic i.i.d. with the popularity distribution $\boldsymbol{\mu}$. When the cache is allowed to update only after every $r$ requests, 
\begin{align*}
   R_{S}^{\textrm{FTPL}(\alpha \sqrt{t}) }  \leq 1+t_0' + 2 \left ( \frac{8}{\Delta_{\min}} + \frac{32 \alpha^2}{\Delta_{\min}} \right ). 
\end{align*}    
where $t_0'=\max \left \{r,\frac{8}{\Delta_{\min}^2} \log \left ( {L^2}\right )  , \frac{32 \alpha^2}{\Delta_{\min}^2 }\log \left ( {L^2}\right ) \right \}$.
\end{theorem}

In  Theorem \ref{thm: constrained_stochastic_main}, we prove an $\mathcal{O}(\max \{r,\log L\})$ upper bound on the regret of the FTPL($\alpha \sqrt{t}$) algorithm under stochastic file requests. While the order-optimality of this policy with respect to $r$ follows from Theorem \ref{thm: restricted_stochastic_main}, we also note that the bound proved here improves upon the worst-case $\mathcal{O}(L^2)$ dependency in the upper bound proved in part (b) of Theorem \ref{thm: restricted_stochastic_main} to $\mathcal{O}(\log L)$. The proof of this result can be found in Appendix \ref{sec: app_constrained_stochastic_ub_ftpl}. 
\subsection{Adversarial requests}
\begin{theorem} \label{thm: restricted_adversarial_main}
Files are requested by an oblivious adversary. When cache updates are restricted to $s+1$ fixed points defined by $\{r_i\}_{i=1}^{s}$ as outlined above,  
\begin{enumerate}[label=(\alph*)]
\item For any online caching policy $\pi$ and $L \geq 2C$, 
\begin{align*}
 R^{\pi}_{A}(T)  \geq \frac{1}{2}   \left(0.15 \sqrt{ C\sum_{i=1}^{s}r_i^2 } \left(1-\frac{(C-1)\left(\sum_{i=1}^{s}r_i^4 \right)}{2\left(\sum_{i=1}^{s}r_i^2 \right)^{2}}\right) -0.6 \, C  \underset{1 \leq i \leq s}{\max} r_i  \right). 
\end{align*}
\item The regret of FTPL($\alpha \sqrt{t}$) is upper bounded as:
\begin{align*}
R_{A}^{\textrm{FTPL}(\alpha \sqrt{t})}(T) &\leq  \mathcal{O}(\alpha \sqrt{T})+ \sqrt{\frac{2}{\pi}} \sum_{i=1}^{s} \frac{r_i^2}{\alpha \sqrt{\sum_{j=0}^{i-1}r_j }}. 
\end{align*}
\end{enumerate}
\end{theorem}
The proof of this theorem can be found in Appendix~\ref{sec: app_restricted_adversarial_lb} and Appendix~\ref{sec: app_restricted_adversarial_ub_ftpl} respectively. In part (a), we prove a lower bound on the regret of any policy $\pi$ under adversarial file requests when cache updates are restricted to $s$ fixed time slots. Note that a necessary condition for this bound to be meaningful is 
\begin{align}
4 \underset{1 \leq i \leq s}{\max} r_i \leq \sqrt{\sum_{i=1}^{s} r_{i}^{2} }. \label{eqn: necessary_condition_restricted_ad_lb}
\end{align}
 When this bound is meaningful, we have that $R_A^{\pi}(T) = \Omega \left (\sqrt{C \sum_{i=1}^{s} r_i^2 } \right )$ for any online caching policy $\pi$. When all the $r_i$'s are equal, this condition translates to $r \leq T/16$. This condition does not hold if any of the $r_i$'s is too large. For instance, when $s=3$ and $r_1=T/2, r_2=r_3=T/4$, this condition does not hold. When $\underset{1 \leq i \leq s}{\max} r_i$ and $\underset{1 \leq i \leq s}{\min} r_i$ are known, a sufficient condition for \eqref{eqn: necessary_condition_restricted_ad_lb} to hold is:
\begin{align*}
\frac{\underset{1 \leq i \leq s}{\max} r_i}{\underset{1 \leq i \leq s}{\min} r_i} \leq \frac{\sqrt{s}}{4}. 
\end{align*}

We now discuss the special case where all $r_i$ are equal to $r$, i.e., $s=T/r$. It follows from part (b) of Theorem \ref{thm: restricted_adversarial_main} that FTPL($\sqrt{rt}$) achieves a regret of $\mathcal{O}(\sqrt{rT})$. The following theorem provides a matching lower bound for this setting that proves that FTPL($\sqrt{rt}$) is order-optimal, the proof of which can be found in Appendix \ref{sec: app_constrained_lb_adversarial}. 
\begin{theorem} \label{thm: constrained_adversarial_main}
Files are requested by an oblivious adversary. When the cache is allowed to update only after every $r$ requests, for any online caching policy $\pi$ and $L \geq 2C$, 
\begin{align*}
  R^{\pi}_{A}(T)  \geq   \begin{cases}
  \sqrt{\frac{CrT}{2\pi }}-\Theta\left({\frac{r\sqrt{r}}{\sqrt{T}}}\right), & \text{when } r=o(T), \\
    \Omega(T), & \text{when } r=\Omega(T). 
  \end{cases}
\end{align*}
\end{theorem}

\section{Numerical experiments} \label{sec: numerical}
In this section, we present the results of numerical simulations for the various policies discussed in Section \ref{sec: policies}. 
\begin{figure}[H]
     \centering
     \begin{subfigure}[b]{0.495\textwidth}
         \centering
         \scalebox{0.4}{ \input{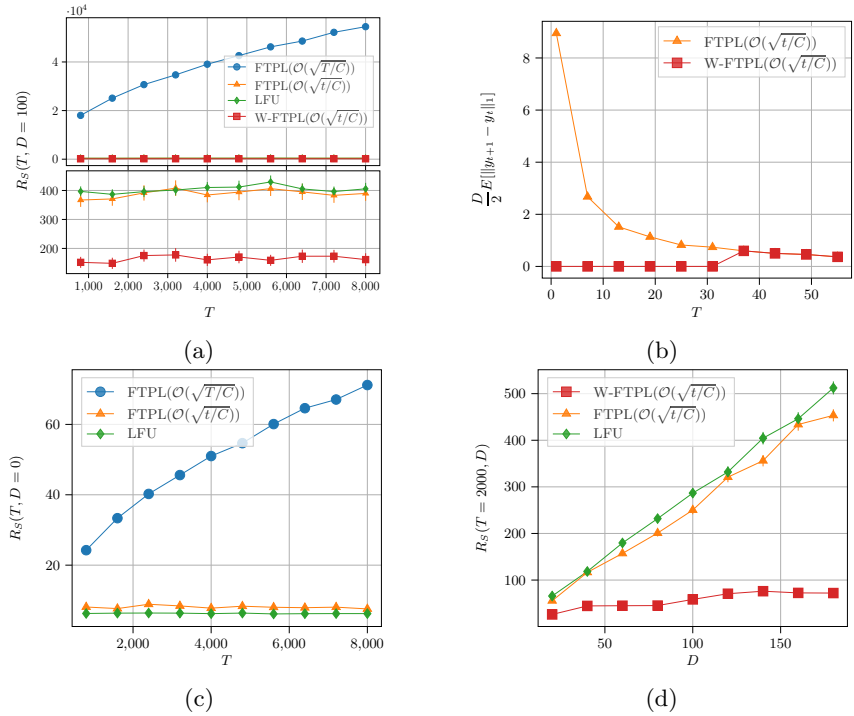} } 
         \subcaption{}
         \label{fig: stochastic_regret}
     \end{subfigure}
     \hfill
     \begin{subfigure}[b]{0.495\textwidth}
         \centering
         \scalebox{0.6}{
\begin{tikzpicture}

\definecolor{darkgray176}{RGB}{176,176,176}
\definecolor{darkorange25512714}{RGB}{255,127,14}
\definecolor{lightgray204}{RGB}{204,204,204}
\definecolor{steelblue31119180}{RGB}{31,119,180}
\definecolor{crimson2143940}{RGB}{214,39,40}

\begin{axis}[
legend cell align={left},
legend style={fill opacity=0.8, draw opacity=1, text opacity=1, draw=lightgray204},
tick align=outside,
tick pos=left,
x grid style={darkgray176},
xlabel={\(\displaystyle T\)},
xmajorgrids,
xmin=-1.7, xmax=57.7,
xtick style={color=black},
y grid style={darkgray176},
ylabel={\(\displaystyle \frac{D}{2} E [\| y_{t+1} - y_{t} \|_1 ]\)},
ymajorgrids,
ymin=-0.44793076810395, ymax=9.40654613018296,
ytick style={color=black}
]
\path [draw=steelblue31119180, semithick]
(axis cs:1,8.93938463792099)
--(axis cs:1,8.95861536207901);

\path [draw=steelblue31119180, semithick]
(axis cs:7,2.66699268545367)
--(axis cs:7,2.68200731454633);

\path [draw=steelblue31119180, semithick]
(axis cs:13,1.5120887699749)
--(axis cs:13,1.5239112300251);

\path [draw=steelblue31119180, semithick]
(axis cs:19,1.12881860218644)
--(axis cs:19,1.13918139781356);

\path [draw=steelblue31119180, semithick]
(axis cs:25,0.817539246683171)
--(axis cs:25,0.826460753316829);

\path [draw=steelblue31119180, semithick]
(axis cs:31,0.733760806805818)
--(axis cs:31,0.742239193194182);

\path [draw=steelblue31119180, semithick]
(axis cs:37,0.593168415054628)
--(axis cs:37,0.600831584945372);

\path [draw=steelblue31119180, semithick]
(axis cs:43,0.495983391290234)
--(axis cs:43,0.503016608709766);

\path [draw=steelblue31119180, semithick]
(axis cs:49,0.455624265421808)
--(axis cs:49,0.462375734578192);

\path [draw=steelblue31119180, semithick]
(axis cs:55,0.364469929380361)
--(axis cs:55,0.370530070619639);

\path [draw=darkorange25512714, semithick]
(axis cs:1,0)
--(axis cs:1,0);

\path [draw=darkorange25512714, semithick]
(axis cs:7,0)
--(axis cs:7,0);

\path [draw=darkorange25512714, semithick]
(axis cs:13,0)
--(axis cs:13,0);

\path [draw=darkorange25512714, semithick]
(axis cs:19,0)
--(axis cs:19,0);

\path [draw=darkorange25512714, semithick]
(axis cs:25,0)
--(axis cs:25,0);

\path [draw=darkorange25512714, semithick]
(axis cs:31,0)
--(axis cs:31,0);

\path [draw=darkorange25512714, semithick]
(axis cs:37,0.593168415054628)
--(axis cs:37,0.600831584945372);

\path [draw=darkorange25512714, semithick]
(axis cs:43,0.495983391290234)
--(axis cs:43,0.503016608709766);

\path [draw=darkorange25512714, semithick]
(axis cs:49,0.455624265421808)
--(axis cs:49,0.462375734578192);

\path [draw=darkorange25512714, semithick]
(axis cs:55,0.364469929380361)
--(axis cs:55,0.370530070619639);

\addplot [semithick, darkorange25512714, mark=triangle*, mark size=3, mark options={solid}]
table {%
1 8.949
7 2.6745
13 1.518
19 1.134
25 0.822
31 0.738
37 0.597
43 0.4995
49 0.459
55 0.3675
};
\addlegendentry{FTPL($\mathcal{O}(\sqrt{t/C})$)}
\addplot [semithick,crimson2143940, mark=square*, mark size=3, mark options={solid}]
table {%
1 0
7 0
13 0
19 0
25 0
31 0
37 0.597
43 0.4995
49 0.459
55 0.3675
};
\addlegendentry{W-FTPL($\mathcal{O}(\sqrt{t/C})$)}
\end{axis}

\end{tikzpicture} }
        \subcaption{}
         \label{fig: switches}
     \end{subfigure} 
     
     \begin{subfigure}[b]{0.495\textwidth}
         \centering
         \scalebox{0.6}{
\begin{tikzpicture}

\definecolor{darkgray176}{RGB}{176,176,176}
\definecolor{darkorange25512714}{RGB}{255,127,14}
\definecolor{forestgreen4416044}{RGB}{44,160,44}
\definecolor{lightgray204}{RGB}{204,204,204}
\definecolor{steelblue31119180}{RGB}{31,119,180}

\begin{axis}[
legend cell align={left},
legend style={
  fill opacity=0.8,
  draw opacity=1,
  text opacity=1,
  at={(0.03,0.97)},
  anchor=north west,
  draw=lightgray204
},
tick align=outside,
tick pos=left,
x grid style={darkgray176},
xlabel={\(\displaystyle T\)},
xmajorgrids,
xmin=440, xmax=8360,
xtick style={color=black},
y grid style={darkgray176},
ylabel={\(\displaystyle R_S(T,D= 0)\)},
ymajorgrids,
ymin=2.56561746509608, ymax=75.5383562410627,
ytick style={color=black}
]
\path [draw=steelblue31119180, semithick]
(axis cs:800,23.6845764382496)
--(axis cs:800,24.8254235617504);

\path [draw=steelblue31119180, semithick]
(axis cs:1600,32.6482556733396)
--(axis cs:1600,34.0517443266604);

\path [draw=steelblue31119180, semithick]
(axis cs:2400,39.4031190173029)
--(axis cs:2400,41.0768809826971);

\path [draw=steelblue31119180, semithick]
(axis cs:3200,44.7763334655335)
--(axis cs:3200,46.4236665344665);

\path [draw=steelblue31119180, semithick]
(axis cs:4000,50.0821857176008)
--(axis cs:4000,51.9278142823992);

\path [draw=steelblue31119180, semithick]
(axis cs:4800,53.7828850757653)
--(axis cs:4800,55.5571149242347);

\path [draw=steelblue31119180, semithick]
(axis cs:5600,59.1208211232182)
--(axis cs:5600,61.0591788767818);

\path [draw=steelblue31119180, semithick]
(axis cs:6400,63.5652853507813)
--(axis cs:6400,65.6447146492187);

\path [draw=steelblue31119180, semithick]
(axis cs:7200,66.1186920368251)
--(axis cs:7200,68.0213079631749);

\path [draw=steelblue31119180, semithick]
(axis cs:8000,70.1585864305721)
--(axis cs:8000,72.2214135694279);

\path [draw=darkorange25512714, semithick]
(axis cs:800,7.39243856980608)
--(axis cs:800,8.81756143019392);

\path [draw=darkorange25512714, semithick]
(axis cs:1600,6.90413596778256)
--(axis cs:1600,8.41586403221743);

\path [draw=darkorange25512714, semithick]
(axis cs:2400,8.05535407315437)
--(axis cs:2400,9.72464592684563);

\path [draw=darkorange25512714, semithick]
(axis cs:3200,7.56989246707889)
--(axis cs:3200,9.26010753292111);

\path [draw=darkorange25512714, semithick]
(axis cs:4000,7.05125841728848)
--(axis cs:4000,8.51874158271152);

\path [draw=darkorange25512714, semithick]
(axis cs:4800,7.52755165973296)
--(axis cs:4800,9.14244834026704);

\path [draw=darkorange25512714, semithick]
(axis cs:5600,7.1765657185231)
--(axis cs:5600,8.8834342814769);

\path [draw=darkorange25512714, semithick]
(axis cs:6400,7.20519124920302)
--(axis cs:6400,8.61480875079698);

\path [draw=darkorange25512714, semithick]
(axis cs:7200,7.265824871383)
--(axis cs:7200,8.864175128617);

\path [draw=darkorange25512714, semithick]
(axis cs:8000,6.90328915212156)
--(axis cs:8000,8.22671084787844);

\path [draw=forestgreen4416044, semithick]
(axis cs:800,6.04132050116294)
--(axis cs:800,6.48867949883706);

\path [draw=forestgreen4416044, semithick]
(axis cs:1600,6.10051708145893)
--(axis cs:1600,6.62948291854107);

\path [draw=forestgreen4416044, semithick]
(axis cs:2400,6.15902645953412)
--(axis cs:2400,6.63097354046587);

\path [draw=forestgreen4416044, semithick]
(axis cs:3200,6.11469583663071)
--(axis cs:3200,6.61530416336929);

\path [draw=forestgreen4416044, semithick]
(axis cs:4000,5.97513331627189)
--(axis cs:4000,6.46486668372811);

\path [draw=forestgreen4416044, semithick]
(axis cs:4800,6.1223662196062)
--(axis cs:4800,6.6376337803938);

\path [draw=forestgreen4416044, semithick]
(axis cs:5600,5.88256013673093)
--(axis cs:5600,6.39743986326907);

\path [draw=forestgreen4416044, semithick]
(axis cs:6400,5.96379169957128)
--(axis cs:6400,6.45620830042872);

\path [draw=forestgreen4416044, semithick]
(axis cs:7200,5.97721063210443)
--(axis cs:7200,6.49278936789558);

\path [draw=forestgreen4416044, semithick]
(axis cs:8000,6.00070931811345)
--(axis cs:8000,6.45929068188655);

\addplot [semithick, steelblue31119180, mark=*, mark size=3, mark options={solid}]
table {%
800 24.255
1600 33.35
2400 40.24
3200 45.6
4000 51.005
4800 54.67
5600 60.09
6400 64.605
7200 67.07
8000 71.19
};
\addlegendentry{FTPL($\mathcal{O}(\sqrt{T/C})$)}
\addplot [semithick, darkorange25512714, mark=triangle*, mark size=3, mark options={solid}]
table {%
800 8.105
1600 7.66
2400 8.89
3200 8.415
4000 7.785
4800 8.335
5600 8.03
6400 7.91
7200 8.065
8000 7.565
};
\addlegendentry{FTPL($\mathcal{O}(\sqrt{t/C})$)}
\addplot [semithick, forestgreen4416044, mark=diamond*, mark size=3, mark options={solid}]
table {%
800 6.265
1600 6.365
2400 6.395
3200 6.365
4000 6.22
4800 6.38
5600 6.14
6400 6.21
7200 6.235
8000 6.23
};
\addlegendentry{LFU}
\end{axis}

\end{tikzpicture} }
        \subcaption{}
         \label{fig: stochastic_regret_D_0}
     \end{subfigure}
     \begin{subfigure}[b]{0.495\textwidth}
         \centering
         \scalebox{0.6}{
\begin{tikzpicture}

\definecolor{darkgray176}{RGB}{176,176,176}
\definecolor{darkorange25512714}{RGB}{255,127,14}
\definecolor{forestgreen4416044}{RGB}{44,160,44}
\definecolor{lightgray204}{RGB}{204,204,204}
\definecolor{steelblue31119180}{RGB}{31,119,180}
\definecolor{crimson2143940}{RGB}{214,39,40}

\begin{axis}[
legend cell align={left},
legend style={
  fill opacity=0.8,
  draw opacity=1,
  text opacity=1,
  at={(0.03,0.97)},
  anchor=north west,
  draw=lightgray204
},
tick align=outside,
tick pos=left,
x grid style={darkgray176},
xlabel={$D$},
xmajorgrids,
xmin=12, xmax=188,
xtick style={color=black},
y grid style={darkgray176},
ylabel={\(\displaystyle R_S(T=2000,D)\)},
ymajorgrids,
ymin=0.350839769430529, ymax=551.333456630496,
ytick style={color=black}
]
\path [draw=steelblue31119180, semithick]
(axis cs:20,25.3955041722062)
--(axis cs:20,27.0094958277938);

\path [draw=steelblue31119180, semithick]
(axis cs:40,43.0354926655731)
--(axis cs:40,46.0395073344269);

\path [draw=steelblue31119180, semithick]
(axis cs:60,43.1157969262247)
--(axis cs:60,46.8142030737753);

\path [draw=steelblue31119180, semithick]
(axis cs:80,43.2275661963528)
--(axis cs:80,47.2874338036472);

\path [draw=steelblue31119180, semithick]
(axis cs:100,56.2120635800443)
--(axis cs:100,60.7429364199557);

\path [draw=steelblue31119180, semithick]
(axis cs:120,68.1717913327902)
--(axis cs:120,73.1482086672098);

\path [draw=steelblue31119180, semithick]
(axis cs:140,73.5604644783152)
--(axis cs:140,78.6395355216848);

\path [draw=steelblue31119180, semithick]
(axis cs:160,70.1885260066786)
--(axis cs:160,74.7214739933215);

\path [draw=steelblue31119180, semithick]
(axis cs:180,70.4975266621867)
--(axis cs:180,73.7024733378133);

\path [draw=darkorange25512714, semithick]
(axis cs:20,53.8391472861771)
--(axis cs:20,57.2608527138229);

\path [draw=darkorange25512714, semithick]
(axis cs:40,112.754538091234)
--(axis cs:40,120.310461908766);

\path [draw=darkorange25512714, semithick]
(axis cs:60,151.55176827185)
--(axis cs:60,162.73823172815);

\path [draw=darkorange25512714, semithick]
(axis cs:80,193.490898254914)
--(axis cs:80,208.169101745086);

\path [draw=darkorange25512714, semithick]
(axis cs:100,241.168353958547)
--(axis cs:100,259.256646041453);

\path [draw=darkorange25512714, semithick]
(axis cs:120,309.89937916293)
--(axis cs:120,331.815620837071);

\path [draw=darkorange25512714, semithick]
(axis cs:140,343.994433257095)
--(axis cs:140,368.275566742905);

\path [draw=darkorange25512714, semithick]
(axis cs:160,420.551264281231)
--(axis cs:160,447.083735718769);

\path [draw=darkorange25512714, semithick]
(axis cs:180,440.493366326254)
--(axis cs:180,466.391633673746);

\path [draw=forestgreen4416044, semithick]
(axis cs:20,63.7322723253068)
--(axis cs:20,67.7827276746932);

\path [draw=forestgreen4416044, semithick]
(axis cs:40,114.478238962928)
--(axis cs:40,122.366761037072);

\path [draw=forestgreen4416044, semithick]
(axis cs:60,173.437380832286)
--(axis cs:60,185.817619167714);

\path [draw=forestgreen4416044, semithick]
(axis cs:80,223.564256849684)
--(axis cs:80,240.215743150316);

\path [draw=forestgreen4416044, semithick]
(axis cs:100,276.312354088575)
--(axis cs:100,296.977645911425);

\path [draw=forestgreen4416044, semithick]
(axis cs:120,320.09484121121)
--(axis cs:120,343.91515878879);

\path [draw=forestgreen4416044, semithick]
(axis cs:140,391.032871329317)
--(axis cs:140,417.907128670683);

\path [draw=forestgreen4416044, semithick]
(axis cs:160,432.087639850155)
--(axis cs:160,460.237360149845);

\path [draw=forestgreen4416044, semithick]
(axis cs:180,498.27620777228)
--(axis cs:180,526.28879222772);

\addplot [semithick, crimson2143940, mark=square*, mark size=3, mark options={solid}]
table {%
20 26.2025
40 44.5375
60 44.965
80 45.2575
100 58.4775
120 70.66
140 76.1
160 72.455
180 72.1
};
\addlegendentry{W-FTPL($\mathcal{O}(\sqrt{t/C})$)}
\addplot [semithick, darkorange25512714, mark=triangle*, mark size=3, mark options={solid}]
table {%
20 55.55
40 116.5325
60 157.145
80 200.83
100 250.2125
120 320.8575
140 356.135
160 433.8175
180 453.4425
};
\addlegendentry{FTPL($\mathcal{O}(\sqrt{t/C})$)}
\addplot [semithick, forestgreen4416044, mark=diamond*, mark size=3, mark options={solid}]
table {%
20 65.7575
40 118.4225
60 179.6275
80 231.89
100 286.645
120 332.005
140 404.47
160 446.1625
180 512.2825
};
\addlegendentry{LFU}
\end{axis}

\end{tikzpicture} }
        \subcaption{}
         \label{fig: stochastic_regret_D}
     \end{subfigure}
        \caption{The plots compare (a) the regret including the switching cost for $D=100$ as a function of $T$, (b) the switching cost in each time slot for $D=30, u=5, \beta=0.6$ as a function of $T$, (c) the regret as a function of $T$ with $D=0$, and (d) the regret including the switching cost for $T=2000$ as a function of $D$, of various caching policies under stochastic file requests. $L=10, C=4$ for each of the plots and the popularity distribution is a dyadic distribution. Parts (a) and (c) show that the regret incurred by FTPL($\mathcal{O}(\sqrt{T})$) is increasing with $T$ while FTPL($\mathcal{O}(\sqrt{t})$), W- FTPL($\mathcal{O}(\sqrt{t})$) and LFU have essentially constant regret. Part (b) shows that  FTPL($\mathcal{O}(\sqrt{t})$) makes more switches at the beginning, thus motivating the W- FTPL($\mathcal{O}(\sqrt{t})$) algorithm. Part (d) shows that while the regret including the switching cost of LFU and FTPL($\mathcal{O}(\sqrt{t})$) increase linearly in $D$, it increases sublinearly for W-FTPL($\mathcal{O}(\sqrt{t})$).}
        \label{fig: stochastic_unlimited}
\end{figure}
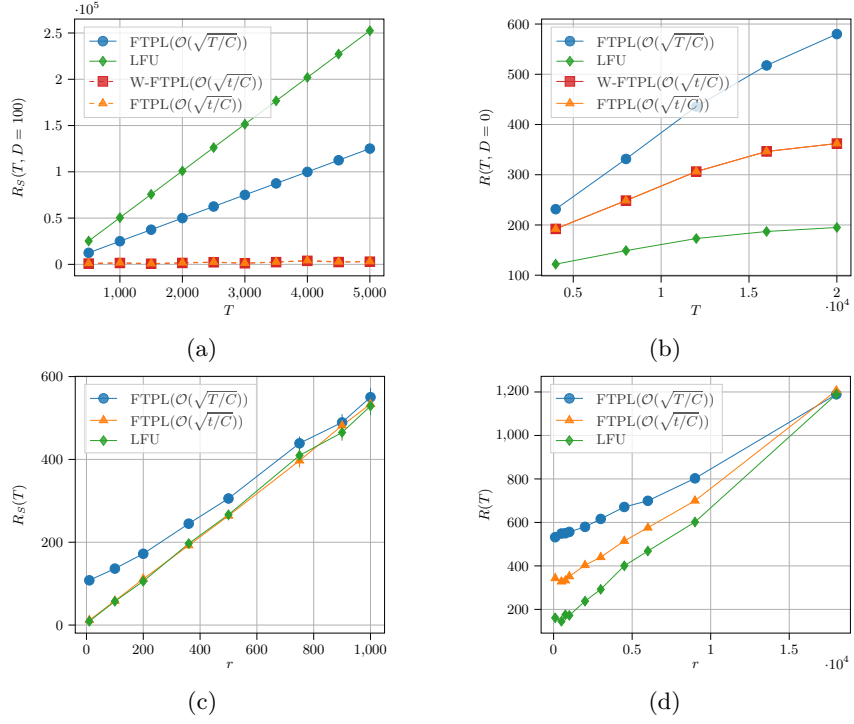
\begin{figure}[H]
	\centering
	\begin{subfigure}[b]{0.495\textwidth}
		\centering
		\scalebox{0.6}{
\begin{tikzpicture}

\definecolor{crimson2143940}{RGB}{214,39,40}
\definecolor{darkgray176}{RGB}{176,176,176}
\definecolor{darkorange25512714}{RGB}{255,127,14}
\definecolor{forestgreen4416044}{RGB}{44,160,44}
\definecolor{lightgray204}{RGB}{204,204,204}
\definecolor{steelblue31119180}{RGB}{31,119,180}

\begin{axis}[
legend cell align={left},
legend style={
  fill opacity=0.8,
  draw opacity=1,
  text opacity=1,
  at={(0.03,0.97)},
  anchor=north west,
  draw=lightgray204
},
tick align=outside,
tick pos=left,
x grid style={darkgray176},
xlabel={\(\displaystyle T\)},
xmajorgrids,
xmin=275, xmax=5225,
xtick style={color=black},
y grid style={darkgray176},
ylabel={\(\displaystyle R_S(T,D= 100)\)},
ymajorgrids,
ymin=-12396.5412539103, ymax=265061.74005971,
ytick style={color=black}
]
\path [draw=steelblue31119180, semithick]
(axis cs:500,12409.4727491008)
--(axis cs:500,12507.4712508992);

\path [draw=steelblue31119180, semithick]
(axis cs:1000,24904.1921336747)
--(axis cs:1000,25043.1798663253);

\path [draw=steelblue31119180, semithick]
(axis cs:1500,37377.71391777)
--(axis cs:1500,37542.64208223);

\path [draw=steelblue31119180, semithick]
(axis cs:2000,49898.3886606836)
--(axis cs:2000,50091.6233393164);

\path [draw=steelblue31119180, semithick]
(axis cs:2500,62397.1013802775)
--(axis cs:2500,62617.2186197225);

\path [draw=steelblue31119180, semithick]
(axis cs:3000,74925.0734742166)
--(axis cs:3000,75164.4865257834);

\path [draw=steelblue31119180, semithick]
(axis cs:3500,87345.4077066943)
--(axis cs:3500,87595.5402933057);

\path [draw=steelblue31119180, semithick]
(axis cs:4000,99761.8698912051)
--(axis cs:4000,100034.182108795);

\path [draw=steelblue31119180, semithick]
(axis cs:4500,112371.021470707)
--(axis cs:4500,112661.742529293);

\path [draw=steelblue31119180, semithick]
(axis cs:5000,124919.743343077)
--(axis cs:5000,125247.712656923);

\path [draw=darkorange25512714, semithick]
(axis cs:500,797.728728463019)
--(axis cs:500,1579.22727153698);

\path [draw=darkorange25512714, semithick]
(axis cs:1000,1043.93078486165)
--(axis cs:1000,2497.35721513835);

\path [draw=darkorange25512714, semithick]
(axis cs:1500,424.868944555826)
--(axis cs:1500,1387.67905544417);

\path [draw=darkorange25512714, semithick]
(axis cs:2000,860.352144939299)
--(axis cs:2000,2765.7598550607);

\path [draw=darkorange25512714, semithick]
(axis cs:2500,1315.23499172141)
--(axis cs:2500,3803.04900827859);

\path [draw=darkorange25512714, semithick]
(axis cs:3000,495.877779462999)
--(axis cs:3000,2390.302220537);

\path [draw=darkorange25512714, semithick]
(axis cs:3500,1048.54728593588)
--(axis cs:3500,4292.74471406411);

\path [draw=darkorange25512714, semithick]
(axis cs:4000,2074.99774791404)
--(axis cs:4000,6219.93825208596);

\path [draw=darkorange25512714, semithick]
(axis cs:4500,1215.75728144227)
--(axis cs:4500,4449.53871855773);

\path [draw=darkorange25512714, semithick]
(axis cs:5000,1172.68966477492)
--(axis cs:5000,5275.15833522508);

\path [draw=forestgreen4416044, semithick]
(axis cs:500,25200)
--(axis cs:500,25200);

\path [draw=forestgreen4416044, semithick]
(axis cs:1000,50450)
--(axis cs:1000,50450);

\path [draw=forestgreen4416044, semithick]
(axis cs:1500,75700)
--(axis cs:1500,75700);

\path [draw=forestgreen4416044, semithick]
(axis cs:2000,100950)
--(axis cs:2000,100950);

\path [draw=forestgreen4416044, semithick]
(axis cs:2500,126200)
--(axis cs:2500,126200);

\path [draw=forestgreen4416044, semithick]
(axis cs:3000,151450)
--(axis cs:3000,151450);

\path [draw=forestgreen4416044, semithick]
(axis cs:3500,176700)
--(axis cs:3500,176700);

\path [draw=forestgreen4416044, semithick]
(axis cs:4000,201950)
--(axis cs:4000,201950);

\path [draw=forestgreen4416044, semithick]
(axis cs:4500,227200)
--(axis cs:4500,227200);

\path [draw=forestgreen4416044, semithick]
(axis cs:5000,252450)
--(axis cs:5000,252450);

\path [draw=crimson2143940, semithick]
(axis cs:500,495.54594458627)
--(axis cs:500,1175.62205541373);

\path [draw=crimson2143940, semithick]
(axis cs:1000,786.026594931995)
--(axis cs:1000,2152.493405068);

\path [draw=crimson2143940, semithick]
(axis cs:1500,215.198805799748)
--(axis cs:1500,1114.56919420025);

\path [draw=crimson2143940, semithick]
(axis cs:2000,607.763789834916)
--(axis cs:2000,2445.07221016508);

\path [draw=crimson2143940, semithick]
(axis cs:2500,1019.04906705999)
--(axis cs:2500,3433.44293294001);

\path [draw=crimson2143940, semithick]
(axis cs:3000,263.591394823591)
--(axis cs:3000,2104.25660517641);

\path [draw=crimson2143940, semithick]
(axis cs:3500,807.629479465825)
--(axis cs:3500,3990.68652053418);

\path [draw=crimson2143940, semithick]
(axis cs:4000,1754.34249239171)
--(axis cs:4000,5826.72550760829);

\path [draw=crimson2143940, semithick]
(axis cs:4500,889.77451604453)
--(axis cs:4500,4061.65348395547);

\path [draw=crimson2143940, semithick]
(axis cs:5000,866.276640156608)
--(axis cs:5000,4912.14335984339);

\addplot [semithick, steelblue31119180, mark=*, mark size=3, mark options={solid}]
table {%
500 12458.472
1000 24973.686
1500 37460.178
2000 49995.006
2500 62507.16
3000 75044.78
3500 87470.474
4000 99898.026
4500 112516.382
5000 125083.728
};
\addlegendentry{FTPL($\mathcal{O}(\sqrt{T/C}))$}
\addplot [semithick, forestgreen4416044, mark=diamond*, mark size=3, mark options={solid}]
table {%
500 25200
1000 50450
1500 75700
2000 100950
2500 126200
3000 151450
3500 176700
4000 201950
4500 227200
5000 252450
};
\addlegendentry{LFU}
\addplot [dashed, crimson2143940, mark=square*, mark size=3, mark options={solid}]
table {%
500 835.584
1000 1469.26
1500 664.884
2000 1526.418
2500 2226.246
3000 1183.924
3500 2399.158
4000 3790.534
4500 2475.714
5000 2889.21
};
\addlegendentry{W-FTPL($\mathcal{O}(\sqrt{t/C}))$}
\addplot [dashed, darkorange25512714, mark=triangle*, mark size=3, mark options={solid}]
table {%
	500 1188.478
	1000 1770.644
	1500 906.274
	2000 1813.056
	2500 2559.142
	3000 1443.09
	3500 2670.646
	4000 4147.468
	4500 2832.648
	5000 3223.924
};
\addlegendentry{FTPL($\mathcal{O}(\sqrt{t/C}))$}
\end{axis}

\end{tikzpicture}}
		\subcaption{}
		\label{fig: adversarial_regret}
	\end{subfigure}
	\hfill
	\begin{subfigure}[b]{0.495\textwidth}
		\centering
		\scalebox{0.6}{
\begin{tikzpicture}

\definecolor{crimson2143940}{RGB}{214,39,40}
\definecolor{darkgray176}{RGB}{176,176,176}
\definecolor{darkorange25512714}{RGB}{255,127,14}
\definecolor{forestgreen4416044}{RGB}{44,160,44}
\definecolor{lightgray204}{RGB}{204,204,204}
\definecolor{steelblue31119180}{RGB}{31,119,180}

\begin{axis}[
legend cell align={left},
legend style={
  fill opacity=0.8,
  draw opacity=1,
  text opacity=1,
  at={(0.03,0.97)},
  anchor=north west,
  draw=lightgray204
},
tick align=outside,
tick pos=left,
x grid style={darkgray176},
xlabel={\(\displaystyle T\)},
xmajorgrids,
xmin=3200, xmax=20800,
xtick style={color=black},
y grid style={darkgray176},
ylabel={\(\displaystyle R(T,D=0)\)},
ymajorgrids,
ymin=98.7661319637941, ymax=609.911228760324,
ytick style={color=black}
]
\path [draw=steelblue31119180, semithick]
(axis cs:4000,228.887370545224)
--(axis cs:4000,233.952629454776);

\path [draw=steelblue31119180, semithick]
(axis cs:8000,327.995688192894)
--(axis cs:8000,334.804311807106);

\path [draw=steelblue31119180, semithick]
(axis cs:12000,431.528474171989)
--(axis cs:12000,439.551525828011);

\path [draw=steelblue31119180, semithick]
(axis cs:16000,511.230546193947)
--(axis cs:16000,523.929453806053);

\path [draw=steelblue31119180, semithick]
(axis cs:20000,573.322639275882)
--(axis cs:20000,586.677360724118);

\path [draw=darkorange25512714, semithick]
(axis cs:4000,187.191674084104)
--(axis cs:4000,197.168325915896);

\path [draw=darkorange25512714, semithick]
(axis cs:8000,241.257236448225)
--(axis cs:8000,255.902763551775);

\path [draw=darkorange25512714, semithick]
(axis cs:12000,297.016387252083)
--(axis cs:12000,315.623612747917);

\path [draw=darkorange25512714, semithick]
(axis cs:16000,335.527674966928)
--(axis cs:16000,357.152325033072);

\path [draw=darkorange25512714, semithick]
(axis cs:20000,350.943992690625)
--(axis cs:20000,373.296007309375);

\path [draw=forestgreen4416044, semithick]
(axis cs:4000,122)
--(axis cs:4000,122);

\path [draw=forestgreen4416044, semithick]
(axis cs:8000,149)
--(axis cs:8000,149);

\path [draw=forestgreen4416044, semithick]
(axis cs:12000,173)
--(axis cs:12000,173);

\path [draw=forestgreen4416044, semithick]
(axis cs:16000,187)
--(axis cs:16000,187);

\path [draw=forestgreen4416044, semithick]
(axis cs:20000,195)
--(axis cs:20000,195);

\path [draw=crimson2143940, semithick]
(axis cs:4000,187.191674084104)
--(axis cs:4000,197.168325915896);

\path [draw=crimson2143940, semithick]
(axis cs:8000,241.257236448225)
--(axis cs:8000,255.902763551775);

\path [draw=crimson2143940, semithick]
(axis cs:12000,297.016387252083)
--(axis cs:12000,315.623612747917);

\path [draw=crimson2143940, semithick]
(axis cs:16000,335.527674966928)
--(axis cs:16000,357.152325033072);

\path [draw=crimson2143940, semithick]
(axis cs:20000,350.943992690625)
--(axis cs:20000,373.296007309375);

\addplot [semithick, steelblue31119180, mark=*, mark size=3, mark options={solid}]
table {%
4000 231.42
8000 331.4
12000 435.54
16000 517.58
20000 580
};
\addlegendentry{FTPL($\mathcal{O}(\sqrt{T/C}))$}
\addplot [semithick, forestgreen4416044, mark=diamond*, mark size=3, mark options={solid}]
table {%
4000 122
8000 149
12000 173
16000 187
20000 195
};
\addlegendentry{LFU}
\addplot [semithick, crimson2143940, mark=square*, mark size=3, mark options={solid}]
table {%
4000 192.18
8000 248.58
12000 306.32
16000 346.34
20000 362.12
};
\addlegendentry{W-FTPL($\mathcal{O}(\sqrt{t/C}))$}
\addplot [semithick, darkorange25512714, mark=triangle*, mark size=3, mark options={solid}]
table {%
	4000 192.18
	8000 248.58
	12000 306.32
	16000 346.34
	20000 362.12
};
\addlegendentry{FTPL($\mathcal{O}(\sqrt{t/C}))$}
\end{axis}

\end{tikzpicture} } 
		\subcaption{}
		\label{fig: ml_D_0}
	\end{subfigure}

	\begin{subfigure}[b]{0.495\textwidth}
		\centering
		\scalebox{0.6}{
\begin{tikzpicture}

\definecolor{darkgray176}{RGB}{176,176,176}
\definecolor{darkorange25512714}{RGB}{255,127,14}
\definecolor{forestgreen4416044}{RGB}{44,160,44}
\definecolor{lightgray204}{RGB}{204,204,204}
\definecolor{steelblue31119180}{RGB}{31,119,180}

\begin{axis}[
legend cell align={left},
legend style={
  fill opacity=0.8,
  draw opacity=1,
  text opacity=1,
  at={(0.03,0.97)},
  anchor=north west,
  draw=lightgray204
},
tick align=outside,
tick pos=left,
x grid style={darkgray176},
xlabel={\(\displaystyle r\)},
xmajorgrids,
xmin=-39.5, xmax=1049.5,
xtick style={color=black},
y grid style={darkgray176},
ylabel={\(\displaystyle R_S(T)\)},
ymajorgrids,
ymin=-19.4353322514098, ymax=600.623840757471,
ytick style={color=black}
]
\path [draw=steelblue31119180, semithick]
(axis cs:10,106.826999196498)
--(axis cs:10,109.261000803502);

\path [draw=steelblue31119180, semithick]
(axis cs:100,133.063149045663)
--(axis cs:100,139.556850954337);

\path [draw=steelblue31119180, semithick]
(axis cs:200,166.997781148382)
--(axis cs:200,177.318218851618);

\path [draw=steelblue31119180, semithick]
(axis cs:360,236.152800138113)
--(axis cs:360,253.499199861887);

\path [draw=steelblue31119180, semithick]
(axis cs:500,294.284846733554)
--(axis cs:500,316.999153266446);

\path [draw=steelblue31119180, semithick]
(axis cs:750,422.293785003109)
--(axis cs:750,455.454214996891);

\path [draw=steelblue31119180, semithick]
(axis cs:900,469.575119524479)
--(axis cs:900,509.216880475522);

\path [draw=steelblue31119180, semithick]
(axis cs:1000,528.052667106568)
--(axis cs:1000,572.439332893431);

\path [draw=darkorange25512714, semithick]
(axis cs:10,11.1084343994618)
--(axis cs:10,12.2155656005382);

\path [draw=darkorange25512714, semithick]
(axis cs:100,55.9475994836462)
--(axis cs:100,60.6044005163538);

\path [draw=darkorange25512714, semithick]
(axis cs:200,106.912854470354)
--(axis cs:200,115.915145529646);

\path [draw=darkorange25512714, semithick]
(axis cs:360,184.457816582038)
--(axis cs:360,200.214183417962);

\path [draw=darkorange25512714, semithick]
(axis cs:500,252.615066826275)
--(axis cs:500,275.284933173724);

\path [draw=darkorange25512714, semithick]
(axis cs:750,380.662117171891)
--(axis cs:750,413.517882828109);

\path [draw=darkorange25512714, semithick]
(axis cs:900,462.364027921183)
--(axis cs:900,501.051972078817);

\path [draw=darkorange25512714, semithick]
(axis cs:1000,511.336355973999)
--(axis cs:1000,556.111644026001);

\path [draw=forestgreen4416044, semithick]
(axis cs:10,8.74917561263029)
--(axis cs:10,9.36282438736971);

\path [draw=forestgreen4416044, semithick]
(axis cs:100,55.0710690756995)
--(axis cs:100,59.6009309243005);

\path [draw=forestgreen4416044, semithick]
(axis cs:200,101.231761447566)
--(axis cs:200,110.080238552434);

\path [draw=forestgreen4416044, semithick]
(axis cs:360,188.677282169805)
--(axis cs:360,204.966717830195);

\path [draw=forestgreen4416044, semithick]
(axis cs:500,255.326071194268)
--(axis cs:500,277.553928805732);

\path [draw=forestgreen4416044, semithick]
(axis cs:750,393.321071378463)
--(axis cs:750,426.022928621537);

\path [draw=forestgreen4416044, semithick]
(axis cs:900,445.287438381789)
--(axis cs:900,484.812561618211);

\path [draw=forestgreen4416044, semithick]
(axis cs:1000,506.739226107285)
--(axis cs:1000,551.152773892715);

\addplot [semithick, steelblue31119180, mark=*, mark size=3, mark options={solid}]
table {%
10 108.044
100 136.31
200 172.158
360 244.826
500 305.642
750 438.874
900 489.396
1000 550.246
};
\addlegendentry{FTPL($\mathcal{O}(\sqrt{T/C}))$}
\addplot [semithick, darkorange25512714, mark=triangle*, mark size=3, mark options={solid}]
table {%
10 11.662
100 58.276
200 111.414
360 192.336
500 263.95
750 397.09
900 481.708
1000 533.724
};
\addlegendentry{FTPL($\mathcal{O}(\sqrt{t/C}))$}
\addplot [semithick, forestgreen4416044, mark=diamond*, mark size=3, mark options={solid}]
table {%
10 9.056
100 57.336
200 105.656
360 196.822
500 266.44
750 409.672
900 465.05
1000 528.946
};
\addlegendentry{LFU}
\end{axis}

\end{tikzpicture}}
		\subcaption{}
		\label{fig: stochastic_regret_r}
	\end{subfigure} 
	\begin{subfigure}[b]{0.495\textwidth}
		\centering
		\scalebox{0.6}{
\begin{tikzpicture}

\definecolor{darkgray176}{RGB}{176,176,176}
\definecolor{darkorange25512714}{RGB}{255,127,14}
\definecolor{forestgreen4416044}{RGB}{44,160,44}
\definecolor{lightgray204}{RGB}{204,204,204}
\definecolor{steelblue31119180}{RGB}{31,119,180}

\begin{axis}[
legend cell align={left},
legend style={
  fill opacity=0.8,
  draw opacity=1,
  text opacity=1,
  at={(0.03,0.97)},
  anchor=north west,
  draw=lightgray204
},
tick align=outside,
tick pos=left,
x grid style={darkgray176},
xlabel={\(\displaystyle r\)},
xmajorgrids,
xmin=-795, xmax=18895,
xtick style={color=black},
y grid style={darkgray176},
ylabel={\(\displaystyle R(T)\)},
ymajorgrids,
ymin=90.9172379973988, ymax=1271.89887416216,
ytick style={color=black}
]
\path [draw=steelblue31119180, semithick]
(axis cs:100,526.846727983116)
--(axis cs:100,537.153272016884);

\path [draw=steelblue31119180, semithick]
(axis cs:500,542.764417302326)
--(axis cs:500,555.035582697674);

\path [draw=steelblue31119180, semithick]
(axis cs:750,544.368938022119)
--(axis cs:750,555.791061977881);

\path [draw=steelblue31119180, semithick]
(axis cs:1000,550.694713951878)
--(axis cs:1000,561.945286048122);

\path [draw=steelblue31119180, semithick]
(axis cs:2000,572.434401737061)
--(axis cs:2000,586.685598262939);

\path [draw=steelblue31119180, semithick]
(axis cs:3000,608.890890786604)
--(axis cs:3000,623.349109213396);

\path [draw=steelblue31119180, semithick]
(axis cs:4500,661.251305211975)
--(axis cs:4500,680.988694788025);

\path [draw=steelblue31119180, semithick]
(axis cs:6000,688.936169603821)
--(axis cs:6000,709.863830396179);

\path [draw=steelblue31119180, semithick]
(axis cs:9000,790.309278144158)
--(axis cs:9000,815.290721855842);

\path [draw=steelblue31119180, semithick]
(axis cs:18000,1178.36650146468)
--(axis cs:18000,1200.27349853532);

\path [draw=darkorange25512714, semithick]
(axis cs:100,333.032402704764)
--(axis cs:100,352.967597295236);

\path [draw=darkorange25512714, semithick]
(axis cs:500,317.381462604162)
--(axis cs:500,338.138537395838);

\path [draw=darkorange25512714, semithick]
(axis cs:750,323.334135496885)
--(axis cs:750,342.225864503115);

\path [draw=darkorange25512714, semithick]
(axis cs:1000,341.552548819769)
--(axis cs:1000,361.967451180231);

\path [draw=darkorange25512714, semithick]
(axis cs:2000,394.630166399006)
--(axis cs:2000,411.849833600994);

\path [draw=darkorange25512714, semithick]
(axis cs:3000,428.795582588517)
--(axis cs:3000,451.244417411483);

\path [draw=darkorange25512714, semithick]
(axis cs:4500,504.426877068115)
--(axis cs:4500,525.173122931885);

\path [draw=darkorange25512714, semithick]
(axis cs:6000,566.786475980396)
--(axis cs:6000,584.773524019604);

\path [draw=darkorange25512714, semithick]
(axis cs:9000,691.444303670187)
--(axis cs:9000,708.395696329813);

\path [draw=darkorange25512714, semithick]
(axis cs:18000,1194.98210929988)
--(axis cs:18000,1218.21789070012);

\path [draw=forestgreen4416044, semithick]
(axis cs:100,160.982912903767)
--(axis cs:100,161.457087096233);

\path [draw=forestgreen4416044, semithick]
(axis cs:500,144.598221459433)
--(axis cs:500,146.001778540567);

\path [draw=forestgreen4416044, semithick]
(axis cs:750,174.520335592575)
--(axis cs:750,176.159664407425);

\path [draw=forestgreen4416044, semithick]
(axis cs:1000,171.173976324304)
--(axis cs:1000,173.386023675696);

\path [draw=forestgreen4416044, semithick]
(axis cs:2000,236.118445963439)
--(axis cs:2000,239.561554036561);

\path [draw=forestgreen4416044, semithick]
(axis cs:3000,289.461044637092)
--(axis cs:3000,294.938955362908);

\path [draw=forestgreen4416044, semithick]
(axis cs:4500,396.961906728383)
--(axis cs:4500,403.878093271617);

\path [draw=forestgreen4416044, semithick]
(axis cs:6000,463.650163320702)
--(axis cs:6000,473.149836679297);

\path [draw=forestgreen4416044, semithick]
(axis cs:9000,595.173619481011)
--(axis cs:9000,608.466380518989);

\path [draw=forestgreen4416044, semithick]
(axis cs:18000,1178.36650146468)
--(axis cs:18000,1200.27349853532);

\addplot [semithick, steelblue31119180, mark=*, mark size=3, mark options={solid}]
table {%
100 532
500 548.9
750 550.08
1000 556.32
2000 579.56
3000 616.12
4500 671.12
6000 699.4
9000 802.8
18000 1189.32
};
\addlegendentry{FTPL($\mathcal{O}(\sqrt{T/C}))$}
\addplot [semithick, darkorange25512714, mark=triangle*, mark size=3, mark options={solid}]
table {%
100 343
500 327.76
750 332.78
1000 351.76
2000 403.24
3000 440.02
4500 514.8
6000 575.78
9000 699.92
18000 1206.6
};
\addlegendentry{FTPL($\mathcal{O}(\sqrt{t/C}))$}
\addplot [semithick, forestgreen4416044, mark=diamond*, mark size=3, mark options={solid}]
table {%
100 161.22
500 145.3
750 175.34
1000 172.28
2000 237.84
3000 292.2
4500 400.42
6000 468.4
9000 601.82
18000 1189.32
};
\addlegendentry{LFU}
\end{axis}

\end{tikzpicture} }
		\subcaption{}
		\label{fig: ml_r_learning_rate_const}
	\end{subfigure}
	\caption{The plots compare (a) the regret including the switching cost for $D=100$ as a function of $T$ on a round robin request sequence, (b) the regret without the switching cost, i.e., $D=0$ as a function of $T$ on the MovieLens dataset, (c) the regret as a function of the constant switching frequency $r$ under stochastic file requests from a dyadic distribution, and (d) the regret as a function of the switching frequency $r$ on the MovieLens dataset, of various caching policies. We used $L=2, C=1$ for Part (a) and $L=10,C=4$ for Part (c). In Part (a), note that the regret including the switching cost scales linearly with $T$ for LFU, while W-FTPL($\mathcal{O}(\sqrt{t})$) and FTPL($\mathcal{O}(\sqrt{t})$) show better performance. In Part (c), the regret scales linearly with $r$ for all the three algorithms as expected. The regret scales sublinearly with $T$ in Part (b) and linearly with $r$ in Part (d).}
	\label{fig: adversarial_and_restricted}
\end{figure}
\subsection{Setting 1 with stochastic file requests}
\emph{Setup.} We use $L=10, C=4$ throughout this section. The popularity distribution used is a dyadic distribution, i.e., for $1 \leq i \le L-1$, $\mu(i)=\frac{1}{2^i}$ and $\mu(L)=\frac{1}{2^{L-1}}$.

\noindent\emph{Results.} Figure~\ref{fig: stochastic_regret} shows that the regret of FTPL including the switching cost increases with $T$, while it is essentially constant for FTPL($\mathcal{O}(\sqrt{t})$), W-FTPL($\mathcal{O}(\sqrt{t})$) and LFU. One can also observe that W- FTPL($\mathcal{O}(\sqrt{t})$) performs the best among all the four algorithms. In Figure~\ref{fig: stochastic_regret_D_0}, we plot only the regret as a function of $T$. Note that the same trend is observed here as well. Here, we omit plotting W- FTPL($\mathcal{O}(\sqrt{t})$) as its regret would be the same as that of FTPL($\mathcal{O}(\sqrt{t})$) in this case. Figure~\ref{fig: switches} shows that FTPL($\mathcal{O}(\sqrt{t})$) indeed makes more switches at the beginning, which is the motivation for the W-FTPL($\mathcal{O}(\sqrt{t})$) policy where no switches are made for an initial period. Figure~\ref{fig: stochastic_regret_D} shows that the regret of LFU and FTPL($\mathcal{O}(\sqrt{t})$) grows linearly with $D$, while that of W-FTPL($\mathcal{O}(\sqrt{t})$) grows sublinearly with $D$.
\subsection{Setting 1 with adversarial file requests}
\emph{Setup.} We consider a synthetic adversarial request sequence in Figure~\ref{fig: adversarial_regret} and a real-world trace in Figure~\ref{fig: ml_D_0}. The synthetic adversarial request sequence used is $1,2,1,2,\ldots$ for $L=2,C=1$, i.e., a round-robin request sequence. The real-world trace used is the first 20,000 rows of the MovieLens 1M dataset \cite{ml1m_website,ML1m_paper} which contains ratings for 2569 movies with timestamps that we model as requests to a CDN server of library size 2569 and a cache size of 25. \\
\noindent \emph{Results.} Figure~\ref{fig: adversarial_regret} shows that under the round-robin request sequence, the regret including the switching cost of LFU scales linearly with $T$, while that of W-FTPL($\mathcal{O}(\sqrt{t})$) and FTPL($\mathcal{O}(\sqrt{t})$) scales sublinearly with $T$. Figure~\ref{fig: ml_D_0} shows that on the MovieLens dataset, the regret scales sublinearly with $T$ for all the four algorithms.
\subsection{Setting 2}
\emph{Setup.} We consider stochastic file requests drawn from a dyadic distribution for $L=10, C=4$ in Figure~\ref{fig: stochastic_regret_r} and file requests from the MovieLens dataset in Figure~\ref{fig: ml_r_learning_rate_const}. We also used $T=18000$ and chose $r$ to be factors of $T$. There were 2518 unique movies in the first 18000 rows of the MovieLens dataset and we set the cache size to be 25.

\noindent \emph{Results.} Figure~\ref{fig: stochastic_regret_r} shows that when file requests are drawn from a dyadic popularity distribution, the regret of all three policies vary linearly with $r$. Figure~\ref{fig: ml_r_learning_rate_const} shows that on the MovieLens dataset, the regret scales linearly with $r$ for all the three policies. 
\section{Conclusion}
We have shown that FTPL($\mathcal{O}( \sqrt{t})$) achieves order-optimal regret even after including the switching cost under stochastic requests. Combining with prior results on the performance of FTPL, it is simultaneously order-optimal for both stochastic and adversarial requests. We also showed that FTPL($\eta$) cannot possibly achieve order-optimal regret simultaneously under both stochastic and adversarial requests, while variants of this policy can individually be order-optimal under each type of request. We proposed the W-FTPL($\mathcal{O}( \sqrt{t})$) policy as a way of preventing the high switching cost incurred by FTPL at the beginning under stochastic file requests. We also considered the restricted switching setting, where the cache is allowed to update its contents only at specific pre-determined time slots and obtained a lower bound on the regret incurred by any policy. We proved an upper bound on the regret of FTPL($\mathcal{O}( \sqrt{t})$) and showed that it is order-optimal under stochastic file requests and in the homogenous restricted switching case under adversarial file requests. 

This work motivates several directions for future work: (1) Bringing the upper and lower bounds closer in the general restricted switching setting would help in proving whether FTPL($\mathcal{O}( \sqrt{t})$) is order-optimal or not under adversarial file requests in this case too. (2) For the restricted switching setting, our results consider only the regret. Adding the switching cost too here would make the bounds complete. 
\bibliography{ref}
\newpage
\appendix
\section{Proof of part (d) of Theorem \ref{thm: unlimited_adversarial}} \label{sec: app_wftpl_adversarial_ub}
In this section, we prove an upper bound on the regret of W-FTPL($\alpha \sqrt{t}$) under adversarial requests. The proof mostly follows that of Theorem 4.1 of \cite{sqrt_t} and is given here for completeness. We bound the regret till time $t'$ by $t'$ and bound the regret incurred from time $t'+1$ in a manner similar to \cite{sqrt_t}. Thus,
\begin{align*}
R_{A}^{\textrm{W-FTPL}(\alpha \sqrt{t})}(T) &= \max _{\boldsymbol{y} \in \mathcal{Y}}\left\langle\boldsymbol{y}, \boldsymbol{X}_{T+1}\right\rangle - \sum_{t=1}^{T} \mathbb{E}_{\boldsymbol{\gamma}}\left[\left\langle\boldsymbol{y}_{t}, \boldsymbol{x}_{t}\right\rangle\right]  \\
&\leq t' +  \max _{\boldsymbol{y} \in \mathcal{Y}}\left\langle\boldsymbol{y}, \boldsymbol{X}_{T+1} - X_{t'+1} \right\rangle - \sum_{t=t'+1}^{T} \mathbb{E}_{\boldsymbol{\gamma}}\left[\left\langle\boldsymbol{y}_{t}, \boldsymbol{x}_{t}\right\rangle\right].  \numberthis \label{eqn: wftpl_adversarial_regret} \\
\end{align*}
We define the potential function $\Phi_{t} : \mathbb{R}^{L} \rightarrow \mathbb{R}$ for all time instants $t$ in the following way:
\begin{align*}
\Phi_{t}(\boldsymbol{x})=\mathbb{E}_{\boldsymbol{\gamma}}\left[\max _{\boldsymbol{y} \in \mathcal{Y}}\left\langle\boldsymbol{y}, \boldsymbol{x}+\eta_{t} \boldsymbol{\gamma}\right\rangle\right],
\end{align*}
 where $\mathcal{Y}$ is the set of possible cache configurations, i.e., the set $\{y \in \{0,1\}^{L}: \|y\|_1 \leq C \}$. As shown in the proof of Proposition 4.1 in \cite{sqrt_t}, we have
 \begin{align*}
\mathbb{E}_{\boldsymbol{\gamma}}\left[\left\langle\boldsymbol{y}_{t}, \boldsymbol{x}_{t}\right\rangle\right] =\Phi_{t}\left(\boldsymbol{X}_{t+1}\right)-\Phi_{t}\left(\boldsymbol{X}_{t}\right)-\frac{1}{2}\left\langle\boldsymbol{x}_{t}, \nabla^{2} \Phi_{t}\left(\widetilde{\boldsymbol{X}}_{t}\right) \boldsymbol{x}_{t}. \right\rangle 
 \end{align*}
where $\widetilde{\boldsymbol{X}}_{t}=\boldsymbol{X}_{t}+\theta_{t} \boldsymbol{x}_{t}$, for some $\theta_{t} \in[0,1]$. Adding this from time $t'+1$ till $T$ gives us:
\begin{align*}
&\sum_{t=t'+1}^{T} \mathbb{E}_{\boldsymbol{\gamma}}\left[\left\langle\boldsymbol{y}_{t}, \boldsymbol{x}_{t}\right\rangle\right] \\
&=\sum_{t=t'+1}^{T}\left[\Phi_{t}\left(\boldsymbol{X}_{t+1}\right)-\Phi_{t}\left(\boldsymbol{X}_{\boldsymbol{t}}\right)\right]-\frac{1}{2} \sum_{t=t'+1}^{T}\left\langle\boldsymbol{x}_{t}, \nabla^{2} \Phi_{t}\left(\widetilde{\boldsymbol{X}}_{t}\right) \boldsymbol{x}_{t}\right\rangle \\
&=\Phi_{T}\left(\boldsymbol{X}_{T+1}\right)-\Phi_{t'+1}\left(\boldsymbol{X}_{t'+1}\right)+\sum_{t=t'+2}^{T}\left[\Phi_{t-1}\left(\boldsymbol{X}_{t}\right)-\Phi_{t}\left(\boldsymbol{X}_{t}\right)\right] \\
&- \frac{1}{2} \sum_{t=t'+1}^{T}\left\langle\boldsymbol{x}_{t}, \nabla^{2} \Phi_{t}\left(\widetilde{\boldsymbol{X}}_{t}\right) \boldsymbol{x}_{t}\right\rangle.
\end{align*}
Substituting this in \eqref{eqn: wftpl_adversarial_regret} gives us:
\begin{align*}
R_{A}^{\textrm{W-FTPL}(\alpha \sqrt{t})}(T) &\leq t' +  \max _{\boldsymbol{y} \in \mathcal{Y}}\left\langle\boldsymbol{y}, \boldsymbol{X}_{T+1} - X_{t'+1} \right\rangle - \Phi_{T}\left(\boldsymbol{X}_{T+1}\right) +\Phi_{t'+1}\left(\boldsymbol{X}_{t'+1}\right) \\ &+\sum_{t=t'+1}^{T-1}\left[\Phi_{t+1}\left(\boldsymbol{X}_{t+1}\right)-\Phi_{t}\left(\boldsymbol{X}_{t+1}\right)\right] +  \frac{1}{2} \sum_{t=t'+1}^{T}\left\langle\boldsymbol{x}_{t}, \nabla^{2} \Phi_{t}\left(\widetilde{\boldsymbol{X}}_{t}\right) \boldsymbol{x}_{t}\right\rangle. \numberthis \label{eqn: wftpl_ad_regret_ub}
\end{align*}
Using Jensen's inequality, we have that
\begin{align*}
\Phi_{T}\left(\boldsymbol{X}_{T+1}\right ) &= \mathbb{E}_{\boldsymbol{\gamma}}\left[\max _{\boldsymbol{y} \in \mathcal{Y}}\left\langle\boldsymbol{y}, \boldsymbol{X_{T+1}}+\eta_{T} \boldsymbol{\gamma}\right\rangle\right] \\
&\geq \max _{\boldsymbol{y} \in \mathcal{Y}} \mathbb{E}_{\boldsymbol{\gamma}}\left[\left\langle\boldsymbol{y}, \boldsymbol{X_{T+1}}+\eta_{T} \boldsymbol{\gamma}\right\rangle\right] \\
&=  \max _{\boldsymbol{y} \in \mathcal{Y}} \left\langle\boldsymbol{y}, \boldsymbol{X_{T+1}} \right\rangle. 
\end{align*}

Substituting the above in \eqref{eqn: wftpl_ad_regret_ub} gives us:
\begin{align*}
R_{A}^{\textrm{W-FTPL}(\alpha \sqrt{t})}(T) &\leq  t' +  \max _{\boldsymbol{y} \in \mathcal{Y}}\left\langle\boldsymbol{y}, \boldsymbol{X}_{T+1} - X_{t'+1} \right\rangle - \max _{\boldsymbol{y} \in \mathcal{Y}} \left\langle\boldsymbol{y}, \boldsymbol{X_{T+1}} \right\rangle  +\Phi_{t'+1}\left(\boldsymbol{X}_{t'+1}\right) \\ &-\sum_{t=t'+2}^{T}\left[\Phi_{t-1}\left(\boldsymbol{X}_{t}\right)-\Phi_{t}\left(\boldsymbol{X}_{t}\right)\right] +  \frac{1}{2} \sum_{t=t'+1}^{T}\left\langle\boldsymbol{x}_{t}, \nabla^{2} \Phi_{t}\left(\widetilde{\boldsymbol{X}}_{t}\right) \boldsymbol{x}_{t}\right\rangle \\ 
&\leq t' + \Phi_{t'+1}\left(\boldsymbol{X}_{t'+1}\right) \\
&+\sum_{t=t'+1}^{T-1}\left[\Phi_{t+1}\left(\boldsymbol{X}_{t+1}\right)-\Phi_{t}\left(\boldsymbol{X}_{t+1}\right)\right]+  \frac{1}{2} \sum_{t=t'+1}^{T}\left\langle\boldsymbol{x}_{t}, \nabla^{2} \Phi_{t}\left(\widetilde{\boldsymbol{X}}_{t}\right) \boldsymbol{x}_{t}\right\rangle, \numberthis \label{eqn: wftpl_ad_regret_ub_maxvalue}
\end{align*}
as $\displaystyle \max _{\boldsymbol{y} \in \mathcal{Y}}\left\langle\boldsymbol{y}, \boldsymbol{X}_{T+1} - \boldsymbol{X}_{t'+1} \right\rangle \leq \max _{\boldsymbol{y} \in \mathcal{Y}} \left\langle\boldsymbol{y}, \boldsymbol{X_{T+1}} \right\rangle$. 
We also have:
\begin{align*}
 \Phi_{t'+1}\left(\boldsymbol{X}_{t'+1}\right) &= \mathbb{E}_{\boldsymbol{\gamma}}\left[\max _{\boldsymbol{y} \in \mathcal{Y}}\left\langle\boldsymbol{y}, \boldsymbol{X}_{t'+1}+\eta_{t'+1} \boldsymbol{\gamma}\right\rangle\right] \\
 &\leq \mathbb{E}_{\boldsymbol{\gamma}}\left[\max _{\boldsymbol{y} \in \mathcal{Y}}\left\langle\boldsymbol{y}, \boldsymbol{X}_{t'+1} \right\rangle\right] +  \mathbb{E}_{\boldsymbol{\gamma}}\left[\max _{\boldsymbol{y} \in \mathcal{Y}} \left\langle \boldsymbol{y}, \eta_{t'+1} \boldsymbol{\gamma}\right\rangle\right] \\
 &\leq t' + \mathbb{E}_{\boldsymbol{\gamma}}\left[\max _{\boldsymbol{y} \in \mathcal{Y}}\left\langle \boldsymbol{y}, \eta_{t'+1} \boldsymbol{\gamma}\right\rangle\right],
\end{align*}
 as $\displaystyle\left\langle\boldsymbol{y}, \boldsymbol{X}_{t'+1} \right\rangle \leq t'$ for $\boldsymbol{y} \in \mathcal{Y}$. Substituting this back in \eqref{eqn: wftpl_ad_regret_ub_maxvalue}, we get
\begin{align*}
R_{A}^{\textrm{W-FTPL}(\alpha \sqrt{t})}(T) &\leq 2t' + \eta_{t'+1}  \mathbb{E}_{\boldsymbol{\gamma}}\left[\max _{\boldsymbol{y} \in \mathcal{Y}}\left\langle \boldsymbol{y}, \boldsymbol{\gamma}\right\rangle\right] \\
&+\sum_{t=t'+1}^{T-1}\left[\Phi_{t+1}\left(\boldsymbol{X}_{t+1}\right)-\Phi_{t}\left(\boldsymbol{X}_{t+1}\right)\right] +  \frac{1}{2} \sum_{t=t'+1}^{T}\left\langle\boldsymbol{x}_{t}, \nabla^{2} \Phi_{t}\left(\widetilde{\boldsymbol{X}}_{t}\right) \boldsymbol{x}_{t}\right\rangle.
\end{align*} 
As shown in \cite{sqrt_t},
\begin{align*}
\Phi_{t+1}\left(\boldsymbol{X}_{t+1}\right)-\Phi_{t}\left(\boldsymbol{X}_{t+1}\right) \leq \left|\eta_{t+1}-\eta_{t}\right| \mathbb{E}_{\boldsymbol{\gamma}}\left[\max _{\boldsymbol{y} \in \mathcal{Y}}\langle\boldsymbol{y}, \gamma\rangle\right].
\end{align*}
Also, \cite{cohenhazan15} proves that: 
\begin{align*}
  \mathbb{E}_{\boldsymbol{\gamma}}\left[\max _{\boldsymbol{y} \in \mathcal{Y}}\left\langle \boldsymbol{y}, \boldsymbol{\gamma}\right\rangle\right ] \leq C \sqrt{2  \log \left ( \frac{N}{C}\right )}.
\end{align*}
It has also been proved in \cite{sigmetrics} that
\begin{align*}
\sum_{t=1}^{T}\left\langle\boldsymbol{x}_{t}, \nabla^{2} \Phi_{t}\left(\widetilde{\boldsymbol{X}}_{t}\right) \boldsymbol{x}_{t}\right\rangle \leq \sqrt{\frac{2}{\pi}} \sum_{t=1}^{T} \frac{1}{\eta_{t}}.
\end{align*}
Combining all these results, we get
\begin{align*}
R_{A}^{\textrm{W-FTPL}(\alpha \sqrt{t})}(T) &\leq 2t' + \eta_{t'+1}C \sqrt{2  \log \left ( \frac{N}{C}\right )}+ \eta_{T} C \sqrt{2 \ln (N e / C)} + \sqrt{\frac{2}{\pi}} \sum_{t=1}^{T} \frac{1}{\eta_{t}}. 
\end{align*}
As W-FTPL$(\alpha \sqrt{t})$ does not incur any switch cost till $t'$ and then incurs the same switch cost as FTPL$(\alpha \sqrt{t})$, the total switch cost till time $T$ incurred by W-FTPL$(\alpha \sqrt{t})$ can be bounded from above by the switch cost of FTPL$(\alpha \sqrt{t})$ proved in Proposition 4.2 of \cite{sqrt_t}, which has been reproduced below:
\begin{align*}
\sum_{t=2}^{T} \mathbb{E}\left[\left\|\boldsymbol{y}_{t+1}-\boldsymbol{y}_{t}\right\|_{1}\right] &\leq \frac{3 \sqrt{2}}{\alpha \sqrt{\pi}}(\sqrt{T}-1)
+(N-1) \frac{2+\sqrt{2 e \ln (2 N)}}{\sqrt{e}} \ln T \\
&+\frac{3(N-1)(2+\sqrt{2 e \ln (2 N)})}{\sqrt{2 \pi e} \alpha}\left(1-T^{-1 / 2}\right) .
\end{align*}
Thus, we get an $\mathcal{O}(\sqrt{T})$ upper bound for $R_{A}^{\textrm{W-FTPL}(\alpha \sqrt{t})}(T,D)$.

\section{Proof of part (b) of Theorem \ref{thm: unlimited_stochastic}} \label{sec: app_ftpl_lb}
In this section, we prove a lower bound on the regret of FTPL with a constant learning rate under stochastic file requests. For any file $k$, let $\hat{\mu}_k(t)=\frac{c_k(t) + \eta_t \gamma(k)}{t}$ and let $\alpha_k(t)$ denote the empirical average number of requests received by file $k$, where $c_k(t)$ denotes the number of requests received by file $k$ at time $t$.
\begin{align*}
\mathbb{E}[R(T)]&=\mathbb{E}\left[\sum_{t=1}^{T} \mathbbm{1}\{x(t) \in \mathcal{C}\}-\mathbbm{1}\{x(t) \in C(t)\} \right ]  \\
& \stackrel{(a)} {=} \mathbb{E}\left[\sum_{t=1}^{T}\left(\sum_{j \in \mathcal{C}} \mu_{j}-\sum_{k \in C(t)} \mu_{k}\right)\right] \\ 
& \geq  \mathbb{E}\left[\sum_{t=1}^{T} \sum_{k \in C(t) \backslash \mathcal{C}}\left(\mu_{C}-\mu_{k}\right)\right] \\
&=\mathbb{E}\left[\sum_{t=1}^{T} \sum_{k \in \mathcal{L} \backslash \mathcal{C}}\left(\mu_{C}-\mu_{k}\right) \mathbbm{1}\{k \in C(t)\}\right] \\
&=\sum_{k \in \mathcal{L} \backslash \mathcal{C}}\left(\mu_{C}-\mu_{k}\right) \mathbb{E}\left[\sum_{t=1}^{T} \mathbbm{1}\{k \in C(t)\}\right] \\
&=\sum_{k \in \mathcal{L} \backslash \mathcal{C}}\left(\mu_{C}-\mu_{k}\right) \sum_{t=1}^{T} \mathbb{P}(k \in C(t)) \\
&\geq (\mu_C-\mu_{C+1}) \sum_{t=1}^{T} \sum_{k \in \mathcal{L} \backslash \mathcal{C}} \mathbb{P}(k \in C(t)) \\
&\geq (\mu_C-\mu_{C+1}) \sum_{t=1}^{T} \mathbb{P} \left ( \bigcup_{k \in \mathcal{L}\setminus \mathcal{C} }k \in C(t) \right ), \numberthis\label{here}
\end{align*}
where $(a)$ follows from the tower property of conditional expectation (i.e., condition on $C(t)$ inside the outer expectation). 
 The event $\left \{\bigcup_{k \in \mathcal{L}\setminus \mathcal{C} }k \in C(t) \right \}$ corresponds to the event where file 2 belongs to $C(t)$, which is equivalent to saying that file 2 has a perturbed count greater than that of file 1  at time $t$. Thus,
\begin{align*}
\mathbb{P} \left ( \bigcup_{k \in \mathcal{L}\setminus \mathcal{C} }k \in C(t) \right ) &= \mathbb{P}(\hat{\mu}_{2}(t)>\hat{\mu}_1(t))\\
&= \mathbb{P} \biggl (  (\alpha_{C+1}(t)-\alpha_{C}(t))-(\mu_{C+1}-\mu_C)  \\&+\frac{\eta}{t}(\gamma_{C+1}(t)-\gamma_{C}(t))>\Delta_{\min} \biggr ).
\end{align*}
Thus, we get 
\begin{multline}\label{lb} 
  \mathbb{P} \left ( \bigcup_{k \in \mathcal{L}\setminus \mathcal{C} }k \in C(t) \right )  \geq \mathbb{P}\left(\alpha_{2}(t)-\alpha_{1}(t)-(\mu_{2}-\mu_1)>-\frac{2}{t}  \right ) \\ \times \mathbb{P} \left (\frac{\eta}{t}\left (\gamma_{2}(t)-\gamma_{1}(t) \right )>\Delta_{\min} + \frac{2}{t} \right),
\end{multline}
 as the perturbation is independent of the file requests seen so far. We also have,
 \begin{align*}
     \mathbb{P}\left (\frac{\eta}{t}(\gamma_{2}(t)-\gamma_{1}(t)>\Delta_{\min} + \frac{2}{t}  \right) &=\mathbb{P} \left ( \mathcal{N}\left (0,2\eta^2 \right ) > t\Delta_{\min}  +2 \right ) \\
     &\geq \frac{1}{4} e^{-(t\Delta_{\min}+2 )^2/4\eta^2}. \numberthis \label{eqn: ftpl_lb_noise}
 \end{align*}
Also,
\begin{align*}
1-\alpha_{2}(t) &= \alpha_1(t). \numberthis \label{eqn: ftpl_lb_estimator}
\end{align*}
Using \eqref{eqn: ftpl_lb_noise} and \eqref{eqn: ftpl_lb_estimator} in \eqref{lb},
\begin{align*}
\mathbb{P} \left ( \bigcup_{k \in \mathcal{L}\setminus \mathcal{C} }k \in C(t) \right ) &\geq \frac{1}{4}\mathbb{P}\left( \alpha_{2}(t) > \frac{1-\Delta_{\min}}{2} -\frac{1}{t} \right)  e^{-(t\Delta_{\min}+2 )^2/4\eta^2}, \numberthis \label{eqn: ftpl_lb_binomial} 
\end{align*} 
where the last step follows from $1-\mu_{1} = \mu_{2}$. The paper \cite{kaas_mean_1980} shows that any median $m$ of a Binomial$(n,p)$ distribution lies in the interval $[\lfloor np \rfloor, \lceil np \rceil]$. The first term in \eqref{eqn: ftpl_lb_binomial} is the probability of a Binomial$(t,\mu_2)$ random variable exceeding $t\mu_2-1 \leq \lfloor t \mu_2 \rfloor$. Thus, we get that 
 \begin{align*}
     \mathbb{P} \left ( \bigcup_{k \in \mathcal{L}\setminus \mathcal{C} }k \in C(t) \right ) &\geq \frac{1}{8} e^{-(t\Delta_{\min}+2 )^2/4\eta^2}. 
 \end{align*}
Adding this over all $T$ gives us:
\begin{align*}
    \mathbb{E}[R(T)] &\geq  \frac{\Delta_{\min}}{8} \sum_{t=1}^{\frac{2\eta}{\Delta_{\min}}} e^{-(t\Delta_{\min}+2 )^2/4\eta^2} \\
    &\geq \frac{ \eta   e^{-\left ( \frac{1+\eta}{\eta} \right )^2}}{4}.
\end{align*}
 
\section{Proof of part (c) of Theorem \ref{thm: unlimited_stochastic}}  \label{sec: app_ftpl_ub_stochastic}
In this section, we prove an upper bound on the regret (including switching cost) of FTPL($\alpha \sqrt{t}$) under stochastic arrivals. We assume $L \geq 3$. To prevent the regret from scaling as a polynomial in $L$, we use the following idea from \cite{hedge}: The algorithm's regret is upper bounded by $t_0$ in the first $t_0$ rounds and in the other $T-t_0$ rounds, the regret is bounded using standard concentration inequalities. The switching cost is upper bounded by $DCt_0$ for the first $t_0$ requests and in the other $T-t_0$ rounds, the switching cost is bounded using standard concentration inequalities. Recall that $t_0=\max \left \{\frac{8}{\Delta_{\min}^2} \log \left ( {L^3}\right )  , \frac{32 \alpha^2}{\Delta_{\min}^2 }\log \left ( {L^3}\right ) \right \}$.

For any file $k$, at time $t$, let $\alpha_k(t)$ denote the empirical average number of requests received by file $k$. We also define, for any file $k$ after $t$ files have been requested from the cache and for $2 \leq i \leq s$,
\begin{align*}
\hat{\mu}_k(t) &:=\frac{c_k(t) + \eta_{t+1} \gamma(k)}{t}, \numberthis \label{eqn: perturbed_average} 
\end{align*}
 where $c_k(t)$ denotes the number of requests received by file $k$ at time $t$. We first consider the switching cost incurred by this policy. Any file can be fetched into the cache as the learning rate varies with time and as the number of switches is twice the number of fetches, the following equation from \cite{sqrt_t} holds:
\begin{align*}
\mathbb{E}\left[\left\|\boldsymbol{y}_{t+1}-\boldsymbol{y}_{t}\right\|_{1}\right] =2 \sum_{f=1}^{L} \mathbb{P}(\text { The file index } f \text { is fetched at time } t+1). 
\end{align*}
The probability of a file $f$ being fetched at time $t+1$ is evaluated by taking the following cases:
\begin{enumerate}
    \item $f \in \mathcal{C}$: \\
    Here, if file $f$ is fetched at time $t+1$, then $f \notin C(t)$ which implies that at time $t$, $\exists f' \in \mathcal{L} \setminus \mathcal{C}$ such that $f' \in C(t)$. This event can be upper bounded using a union bound over files in $\mathcal{L} \setminus \mathcal{C}$ for $f'$.
    \item $f \notin \mathcal{C}$: \\
    Here, if file $f$ is fetched at time $t+1$, then $f \in C(t+1)$ which implies that at time $t+1$, $\exists f' \in \mathcal{C}$ such that $f' \notin C(t+1)$. This event can be upper bounded using a union bound over files in $\mathcal{C}$ for $f'$.
\end{enumerate}
Therefore,
\begin{multline*}
  \sum_{f=1}^{L}   \mathbb{P}(\text { The file index } f \text { is fetched at time } t+1) \leq \sum_{j=1}^{C} \sum_{k=C+1}^{L} \mathbb{P}(\hat{\mu}_k(t-1)>\hat{\mu}_j(t-1)) \\+  \sum_{k=C+1}^{L} \sum_{j=1}^{C}  \mathbb{P}(\hat{\mu}_k(t)>\hat{\mu}_j(t)).
\end{multline*}
Thus the expected number of switches till $T$ can be bounded as:
\begin{align*}
\frac{D}{2} \sum_{t=2}^T \mathbb{E}\left[\left\|\boldsymbol{y}_{t+1}-\boldsymbol{y}_{t}\right\|_{1}\right]  &\leq  DCt_0+\frac{D}{2} \sum_{t=t_0+1}^{T-1} \sum_{j=1}^{C} \sum_{k=C+1}^{L} \mathbb{P}(\hat{\mu}_k(t-1)>\hat{\mu}_j(t-1)) \\&+\frac{D}{2} \sum_{t=t_0+1}^{T-1} \sum_{j=1}^{C} \sum_{k=C+1}^{L}  \mathbb{P}(\hat{\mu}_k(t)>\hat{\mu}_j(t))  \\
&\leq DCt_0+ D\sum_{t=t_0+1}^{T-1} \sum_{j=1}^{C} \sum_{k=C+1}^{L} \mathbb{P}(\hat{\mu}_k(t)>\hat{\mu}_j(t)) \\
&\leq DCt_0 + D \sum_{t=t_0+1}^{T-1} \sum_{j=1}^{C} \sum_{k=C+1}^{L} \mathbb{P}\left( \hat{\mu}_{j}(t)-\mu_{j} \leq-\Delta_{j, k} / 2\right )\\
&+\mathbb{P}\left(\hat{\mu}_{k}(t)-\mu_{k}>\Delta_{j, k} / 2\right) \\
 &\leq DCt_0+D \sum_{t=t_0+1}^{T} \sum_{j=1}^{C} \sum_{k=C+1}^{L} \mathbb{P}\left( \alpha_{j}(t)-\mu_{j} \leq-\Delta_{j, k} / 4\right )\\
 &+\mathbb{P}\left(\alpha_{k}(t)-\mu_{k}>\Delta_{j, k} / 4\right) \\&+\mathbb{P} \left( \eta_t \gamma(j) \leq-t\Delta_{j, k} / 4 \right )
+\mathbb{P} \left( \eta_t \gamma(k) >t\Delta_{j, k} / 4 \right ),
\end{align*}
where the last two steps follow from a union bounding argument. Now, using the Hoeffding inequality \cite{hoeffding}, 
\begin{align*}
\mathbb{P}\left(\alpha_{j}(t)-\mu_{j} \leq-\Delta_{j, k} / 4\right) &\leq e^{-t \Delta_{j, k}^{2} / 8}, \\
\mathbb{P}\left(\alpha_{k}(t)-\mu_{k}>\Delta_{j, k} / 4\right) &\leq e^{-t^{2} \Delta_{j, k}^{2} / 8}.
\end{align*}
We also have:
\begin{align*}
\mathbb{P}\left(\eta_t \gamma(j) \leq-t \Delta_{j, k} / 4\right)&=\mathbb{P}\left(\eta_t \gamma(k)>t \Delta_{j, k} / 4\right) \\&\leq e^{-t^{2} \Delta_{j, k}^{2} / 32 \eta_t^{2}} = e^{-t\Delta_{j, k}^{2} / 32 \alpha^{2}}.
 \end{align*}     
Thus, we get the following bound on the switching cost:
\begin{align}
    \sum_{t=1}^{T-1} \mathbb{E}\left[\left\|\boldsymbol{y}_{t+1}-\boldsymbol{y}_{t}\right\|_{1}\right] \leq DCt_0+ 2D \sum_{t=t_0+1}^{T-1} \sum_{j=1}^{C} \sum_{k=C+1}^{L} e^{-t \Delta_{j, k}^{2} / 8} + e^{-t\Delta_{j, k}^{2} / 32 \alpha^{2}}. \label{eqn: stochastic_eta_t_switchcost}
\end{align}
The regret for the first $t_0$ rounds is bounded by $t_0$. For $t>t_0$, we use ideas from \cite{learning_to_cache} and \cite{hedge} to upper bound the regret.
\begin{align*}
R_S^{\textrm{FTPL}(\alpha \sqrt{t})}(T) &=\mathbb{E}\left[\sum_{t=1}^{T} \mathbbm{1}\{x(t) \in \mathcal{C}\}-\mathbbm{1}\{x(t) \in C(t)\}\right]  \\
&=\mathbb{E}\left[\sum_{t=1}^{T}\left(\sum_{j \in \mathcal{C}} \mu_{j}-\sum_{k \in C(t)} \mu_{k}\right)\right] \\
&\leq t_0+ \mathbb{E}\left[\sum_{t=t_0+1}^{T} \sum_{j=1}^{C} \sum_{k=C+1}^{L} \Delta_{j, k} \mathbbm{1}\{j \notin C(t), k \in C(t)\}\right] \\
&\leq t_0+ \mathbb{E}\left[\sum_{t=t_0+1}^{T} \sum_{j=1}^{C} \sum_{k=C+1}^{L} \Delta_{j, k} \mathbbm{1}\left(\hat{\mu}_{k}(t)>\hat{\mu}_{j}(t)\right)\right] \\
&\leq t_0+ \mathbb{E}\biggl[\sum_{t=t_0+1}^{T} \sum_{j=1}^{C} \sum_{k=C+1}^{L} \Delta_{j, k}\biggl(\mathbbm{1}\left\{\hat{\mu}_{j}(t)-\mu_{j} \leq-\Delta_{j, k} / 2\right\} 
\\&+\mathbbm{1}\left\{\hat{\mu}_{k}(t)-\mu_{k}>\Delta_{j, k} / 2\right\}\biggr)\biggr] \\ 
&\leq t_0+ \mathbb{E}\biggl[\sum_{t=t_0+1}^{T} \sum_{j=1}^{C} \sum_{k=C+1}^{L} \Delta_{j, k} \biggl ( \mathbbm{1}\left\{\alpha_j(t)-\mu_{j} \leq-\Delta_{j, k} / 4\right\} \\&+ \mathbbm{1}\left\{\eta_t\gamma(j)\leq-t\Delta_{j, k} / 4\right\} \biggr ) \biggr] \\
&+\mathbb{E}\biggl[\sum_{t=t_0+1}^{T} \sum_{j=1}^{C} \sum_{k=C+1}^{L} \Delta_{j, k} \biggl ( \mathbbm{1}\left\{\alpha_k(t)-\mu_{k}>\Delta_{j, k} / 4\right\}\ \\&+ \mathbbm{1} \left \{ \eta_t\gamma(k) > t\Delta_{j, k} / 4 \right \} \biggr) \biggr]. \numberthis \label{eqn: eta_t_stochastic_regret_bound}
\end{align*}
By taking expectation inside the summation, we obtain the following upper bound for \eqref{eqn: eta_t_stochastic_regret_bound}:
\begin{align*}
R_S^{\textrm{FTPL}(\alpha \sqrt{t})}(T)  & \leq t_0+ \sum_{j=1}^{C} \sum_{k=C+1}^{L} \sum_{t=t_0+1}^{T} 2 \Delta_{j, k} \left ( e^{-t \Delta_{j, k}^{2} / 8} + e^{-t^2 \Delta_{j,k}^2/32 \eta^2} \right ) \numberthis \label{eqn: eta_t_regret} 
\end{align*}
Thus, combining \eqref{eqn: stochastic_eta_t_switchcost} and \eqref{eqn: eta_t_regret} , we have:
\begin{align*}
R_S^{\textrm{FTPL}(\alpha \sqrt{t})}(T,D) &\leq   (1+DC)t_0 \\&+ 2D\sum_{t=t_0+1}^{T} \sum_{j=1}^{C} \sum_{k=C+1}^{L}   \biggl(e^{-t \Delta_{j, k}^{2} / 8} \\&+e^{-t \Delta_{j, k}^{2} / 32 \alpha^{2}}\biggr) \\&+  2\sum_{t=t_0+1}^{T} \sum_{j=1}^{C} \sum_{k=C+1}^{L}  \Delta_{j,k} \biggl(e^{-t \Delta_{j, k}^{2} / 8} \\&+e^{-t \Delta_{j, k}^{2} / 32 \alpha^{2}}\biggr). \numberthis \label{all}
\end{align*}
Each of these terms are bounded separately in the following way:
\begin{align*}
 \sum_{t=t_0+1}^{T} \sum_{j=1}^{C} \sum_{k=C+1}^{L}   e^{-t \Delta_{j, k}^{2} / 8} &\leq \sum_{t=t_0+1}^{T} C(L-C) e^{-t \Delta_{\min}^2/8} \\& \leq  C(L-C) e^{-t_0 \Delta_{\min}^2/8} \sum_{t=t_0+1}^{T} e^{-(t-t_0) \Delta_{\min}^2/8}\\
 &\leq \sum_{t=1}^{T-t_0} e^{-t \Delta_{\min}^2/8} \leq \frac{8}{\Delta_{\min}^2 }. 
 \end{align*}
 Similarly,
 \begin{align*}
 \sum_{t=t_0+1}^{T} \sum_{j=1}^{C} \sum_{k=C+1}^{L}  e^{-t \Delta_{j, k}^{2} / 32 \alpha^{2}} &\leq \sum_{t=t_0+1}^{T} C(L-C) e^{-t \Delta_{\min}^{2} / 32 \alpha^{2}} \\
 &\leq L^2 e^{-t_0 \Delta_{\min}^{2} / 32 \alpha^{2}}\sum_{t=t_0+1}^{T} e^{-(t-t_0)\Delta_{\min}^{2} / 32 \alpha^{2}}.\\
 &\leq \sum_{t=1}^{T-t_0} e^{-t \Delta_{\min}^{2} / 32 \alpha^{2}} \leq \frac{32 \alpha^2}{\Delta_{\min}^2 }.
  \end{align*}
The function $f(u)=$ $u e^{-u^{2} / 2}$ is decreasing on $[1,+\infty)$. Since $\Delta_{j,k} \geq \Delta_{\min}, j \in \mathcal{C}, k \notin \mathcal{C}$, we get
 \begin{align*}
\Delta_{j,k} e^{-t \Delta_{j,k}^{2} / 8} &=\frac{2}{\sqrt{t}} f\left(\frac{\sqrt{t} \Delta_{j,k}}{2}\right) \\&\leq \frac{2}{\sqrt{t}} f\left(\frac{\sqrt{t} \Delta_{\min}}{2}\right)=\Delta_{\min} e^{-t \Delta_{\min}^{2} / 8} .
  \end{align*}
Note that for $t \geq t_0+1$, $t> \frac{4}{\Delta_{\min}^2}$ holds as $t_0 \geq \frac{8}{\Delta_{\min}^2}$. Using this, we have that:
 \begin{align*}
 \sum_{t=t_0+1}^{T} \sum_{j=1}^{C} \sum_{k=C+1}^{L}  \Delta_{j,k} e^{-t \Delta_{j, k}^{2} / 8} &\leq  \sum_{t=t_0+1}^{T} C(L-C) \Delta_{\min} e^{-t \Delta_{\min}^{2} / 8} &\leq \frac{8}{\Delta_{\min} }.
  \end{align*}
  Similarly,
   \begin{align*}
   \sum_{t=t_0+1}^{T} \sum_{j=1}^{C} \sum_{k=C+1}^{L}  \Delta_{j,k} e^{-t \Delta_{j, k}^{2} / 32 \alpha^{2}}  &\leq \frac{32 \alpha^2}{\Delta_{\min} }. 
\end{align*}
Substituting these bounds in \eqref{eqn: eta_t_regret}, we have the following upper bound on the regret of FTPL($\alpha \sqrt{t}$):
\begin{align*}
R_S^{\textrm{FTPL}(\alpha \sqrt{t})}(T,D)   \leq \left (1+ DC \right ) t_0 + 2\left (1+ \frac{D}{\Delta_{\min}} \right ) \left ( \frac{8}{\Delta_{\min}} + \frac{32 \alpha^2}{\Delta_{\min}} \right )  . 
\end{align*}

\section{Proof of part (d) of Theorem \ref{thm: unlimited_stochastic}}\label{sec: app_wftpl_ub}
In this section, we prove an upper bound on the regret (including the switching cost) of the W-FTPL algorithm under stochastic file requests. Using \eqref{eqn: stochastic_eta_t_switchcost},
\begin{align*}
   \frac{D}{2} \sum_{t=t'+1}^{T} \mathbb{E}\left[\left\|\boldsymbol{y}_{t+1}-\boldsymbol{y}_{t}\right\|_{1}\right] &\leq  2D\sum_{t=t'+1}^{T} \sum_{j=1}^{C} \sum_{k=C+1}^{L}   \left(e^{-t \Delta_{j, k}^{2} / 8}+e^{-t \Delta_{j, k}^{2} / 32 \alpha^{2}}\right) \\
   & \leq 2D C(L-C) \sum_{t=t'+1}^{T}  \left(e^{-t \Delta_{\min}^{2} / 8}+e^{-t \Delta_{\min}^{2} / 32 \alpha^{2}}\right) \\
   & \leq 2D C(L-C) \sum_{t=t'+1}^{T}  \biggl( e^{-t' \Delta_{\min}^{2} / 8} e^{-(t-t') \Delta_{\min}^{2} / 8} \\&+e^{-t' \Delta_{\min}^{2} / 32 \alpha^{2}} e^{-(t-t') \Delta_{\min}^{2} / 32 \alpha^{2}} \biggr) \\
   &\leq 2DC(L-C) \left ( e^{-t' \frac{\Delta_{\min}^{2}}{8} } \frac{8}{\Delta_{\min}^2} + e^{-t' \frac{\Delta_{\min}^{2}}{32 \alpha^{2}}} \frac{32 \alpha^2}{\Delta_{\min}^2} \right ) \\
&\leq 2DC(L-C) \biggl ( e^{- u (\log D)^{1+\beta} \Delta_{\min}^{2} / 8} \frac{8}{\Delta_{\min}^2} \\&+ e^{-u (\log D)^{1+\beta} \Delta_{\min}^{2} / 32 \alpha^{2}} \frac{32 \alpha^2}{\Delta_{\min}^2} \biggr). \numberthis \label{eqn: wftpl_switchcost}
\end{align*}
Till $t'$, the regret is bounded by $t'$. Thus, the regret can be bounded in the following way using \eqref{eqn: eta_t_regret} :
\begin{align*}
    \mathbb{E}\left[\sum_{t=1}^{T} \mathbbm{1}\{x(t) \in \mathcal{C}\}-\mathbbm{1}\{x(t) \in C(t)\} \right ] 
        &\leq t' + 2\sum_{t=t'+1}^T \sum_{j=1}^{C} \sum_{k=C+1}^{L} \Delta_{j,k}\biggl(e^{-t \Delta_{j, k}^{2} / 8} \\&+e^{-t \Delta_{j, k}^{2} / 32 \alpha^{2}}\biggr) \\
        &\leq t' + 2 \left ( \frac{8}{\Delta_{\min }} + \frac{32 \alpha^2}{\Delta_{\min }} \right ). \numberthis \label{eqn: wftpl_regret}
\end{align*}
Combining \eqref{eqn: wftpl_switchcost} and \eqref{eqn: wftpl_regret} gives the following upper bound on the regret of W-FTPL($\alpha \sqrt{t}$):
\begin{align*}
    R_S^{\textrm{W-FTPL}(\alpha \sqrt{t})}(T,D) &\leq t' + \frac{16}{\Delta_{\min }} + \frac{64 \alpha^2}{\Delta_{\min } } + 2LDC(L-C) \biggl ( e^{- u (\log D)^{1+\beta} \Delta_{\min}^{2} / 8} \frac{8}{\Delta_{\min}^2} \\&+ e^{-u (\log D)^{1+\beta} \Delta_{\min}^{2} / 32 \alpha^{2}} \frac{32 \alpha^2}{\Delta_{\min}^2} \biggr ).
\end{align*}
\section{Proof of part (a) of Theorem \ref{thm: restricted_stochastic_main}} \label{sec: app_restricted_stochastic_lb}
In this section, we prove a lower bound on the regret of any policy $\pi$ under stochastic file requests when cache updates are restricted to $s+1$ fixed points. To provide a lower bound on the regret, we first prove a lower bound on the regret for the initial set of $r_1$ requests. 
\begin{align*}
R_{S}^{\pi}(r_1)&\geq \mathbb{E}\left[\sum_{t=1}^{r_1} \mathbbm{1}\{x(t) \in \mathcal{C}\}-\mathbbm{1}\{x(t) \in C(t)\} \right ]  \\
& \stackrel{(a)} {=} \mathbb{E}\left[\sum_{t=1}^{r_1}\left(\sum_{j \in \mathcal{C}} \mu_{j}-\sum_{k \in C(t)} \mu_{k}\right)\right] \\ 
& \geq  \mathbb{E}\left[\sum_{t=1}^{r_1} \sum_{k \in C(t) \backslash \mathcal{C}}\left(\mu_{C}-\mu_{k}\right)\right] \\
&=\mathbb{E}\left[\sum_{t=1}^{r_1} \sum_{k \in \mathcal{L} \backslash \mathcal{C}}\left(\mu_{C}-\mu_{k}\right) \mathbbm{1}\{k \in C(t)\}\right] \\
&=\sum_{k \in \mathcal{L} \backslash \mathcal{C}}\left(\mu_{C}-\mu_{k}\right) \mathbb{E}\left[\sum_{t=1}^{r_1} \mathbbm{1}\{k \in C(t)\}\right] \\
&=\sum_{k \in \mathcal{L} \backslash \mathcal{C}}\left(\mu_{C}-\mu_{k}\right) \sum_{t=1}^{r_1} \mathbb{P}(k \in C(t)) \\
&\geq (\mu_C-\mu_{C+1}) \sum_{t=1}^{r_1} \sum_{k \in \mathcal{L} \backslash \mathcal{C}} \mathbb{P}(k \in C(t)),\numberthis\label{eqn: restricted_stochastic_lb}
\end{align*}
where $(a)$ follows from the tower property of conditional expectation (i.e., condition on $C(t)$ inside the outer expectation). Note that $C(t)$ remains constant during the first $r_1$ requests and consists of $C$ files chosen randomly from the library of $L$ files. Thus, for any $k \in \mathcal{L} \setminus \mathcal{C}$ and $1 \leq t \leq r_1$, 
\begin{align*}
\mathbb{P}(k \in C(t)) = \frac{{N-1 \choose C-1}}{{N \choose C}} = \frac{C}{N}.
\end{align*}
Using this in \eqref{eqn: restricted_stochastic_lb}, we get that
\begin{align*}
R_{S}^{\pi}(r_1) &\geq r_1(\mu_C-\mu_{C+1}) (L-C) C/N. 
\end{align*}
For proving a lower bound on the regret for larger values of $T$, we consider $L=2, C=1$. Given a popularity distribution $\mu=(\mu_1, \mu_2)$ such that $\mu_1 > \mu_2$ and a policy $\pi$, 
\begin{align*}
	R_{S, \mu}^{\pi} (T)  &= \mathbb{E}_{\pi \mu} \left[ \sum_{t=1}^{T} \mathbbm{1}\{x(t) \in \mathcal{C}\}-\mathbbm{1}\{x(t) \in C(t)\} \right] \\
	&= \sum_{t=1}^{T} \left(\mu_1 - \mu_2 \right) \mathbb{E}_{\pi \mu} \left[\mathbbm{1} \left\{C(t) = \{2\} \right\} \right] \\
	&= \left(\mu_1 - \mu_2 \right)  \sum_{t=1}^{T} \mathbb{P}_{\pi \mu} \left ( C(t) = \{2\}  \right ).
\end{align*}
Now, consider the popularity distribution $\mu'=(\mu_2, \mu_1)$. Following the same steps for this popularity distribution, we have that
\begin{align*}
R_{S, \mu'}^{\pi} (T)  &=	\left(\mu_1 - \mu_2 \right)  \sum_{t=1}^{T} \mathbb{P}_{\pi \mu'} \left ( C(t) = \{1\}  \right ).
\end{align*}
Defining $\Delta = \mu_1 - \mu_2$, we have that
\begin{multline*}
	R_{S, \mu'}^{\pi} (T) + R_{S, \mu}^{\pi} (T) = \Delta \sum_{t=1}^{T} \mathbb{P}_{\pi \mu} \left (C(t) = \{2\}  \right ) + \Delta  \sum_{t=1}^{T} \mathbb{P}_{\pi \mu'} \left ( C(t) = \{1\} \right ).
\end{multline*}
Using the fact that the cache configuration remains constant except at the places where the cache is allowed to change its contents, we have that
\begin{multline}
	R_{S, \mu'}^{\pi} (T) + R_{S, \mu}^{\pi} (T) \geq R_{S, \mu'}^{\pi} (r_1) + R_{S, \mu}^{\pi} (r_1) \\+ \Delta  \sum_{i=2}^{s} r_i \Bigl ( \mathbb{P}_{\pi \mu} \left ( C(t_i+1) = \{2\}  \right ) + \mathbb{P}_{\pi \mu'} \left ( C(t_i+1) = \{1\}  \right ) \Bigr ), \label{eqn: restricted_lb_main}
\end{multline}
where $t_i$ has been defined in \eqref{eqn: cum_sums}. 
\begin{lemma}[\cite{lattimore_szepesvari_2020}, Theorem 14.2 (Bretagnolle-Huber inequality)] \label{lem: BR_ineq}
	Let $P$ and $Q$ be probability distributions on the same measurable space $(\Omega, \mathcal{F})$ and let $A \in \mathcal{F}$ be an arbitrary event. Then, 
	\begin{align*}
		P(A) + Q(A^{\mathsf{c}}) \geq \frac{1}{2} \exp \left(-D(P, Q)\right).
	\end{align*}
\end{lemma}
Consider a fixed $i$ such that $2 \leq i \leq s$. Setting $A$ to be the event $\{C(t_i+1) = \{2\} \}$ in Lemma \ref{lem: BR_ineq}, we have that
\begin{align*}
	\mathbb{P}_{\pi \mu} \left ( C(t_i+1) = \{2\} \right )+ \mathbb{P}_{\pi \mu'} \left ( C(t_i+1) = \{1\} \right ) \geq \frac{1}{2} \exp \left(- D(\mathbb{P}_{\pi \mu}, \mathbb{P}_{\pi \mu'})\right).
\end{align*}
\begin{lemma} \label{lem: div}
Let $\mathbb{C}$ be the set of valid caching configurations. If $\mathbb{P}_{\pi \mu}$ and $\mathbb{P}_{\pi \mu'}$ are probability measures on the set $\mathcal{G}_i := \{[2]^{t_i} \times \mathbb{C}^{t_i +1} \}$, then
	\begin{align*}
		D(\mathbb{P}_{\pi \mu}, \mathbb{P}_{\pi \mu'}) = t_i D(\mathbb{P}_{\mu}, \mathbb{P}_{\mu'}),
	\end{align*}
where $\mathbb{P}_{\mu}, \mathbb{P}_{\mu'}$ are the corresponding marginal distributions.
\end{lemma}
Using the above lemma, we have that
\begin{align*}
	\mathbb{P}_{\pi \mu} \left ( C(t_i+1) = \{2\}  \right ) + \mathbb{P}_{\pi \mu'} \left (C(t_i+1) = \{1\} \right ) &\geq \frac{1}{2} \exp \left(- t_i D(\mathbb{P}_{\mu}, \mathbb{P}_{\mu'})\right) \\
	&\geq \frac{1}{2} \exp \left(- t_i \frac{\Delta^2}{\mu_1 \mu_2 } \right),
\end{align*}
where the last step follows from an upper bound on the KL Divergence between two Bernoulli distributions. Substituting this back in \eqref{eqn: restricted_lb_main}, we have that 
\begin{align*}
	R_{S, \mu'}^{\pi} (T) + R_{S, \mu}^{\pi} (T) \geq r_1 \Delta +  \sum_{i=2}^{s} r_i \frac{\Delta}{2} \exp \left(- t_i \frac{\Delta^2}{\mu_1 \mu_2 } \right).
\end{align*}
Thus, 
\begin{align*}
\max \left\{ R_{S, \mu'}^{\pi} (T), R_{S, \mu}^{\pi} (T) \right\} \geq \frac{r_1 \Delta}{2} +  \sum_{i=2}^{s} r_i \frac{\Delta}{4} \exp \left(- t_i \frac{\Delta^2}{\mu_1 \mu_2 } \right),
\end{align*}
which gives us the result. 
\begin{proof}[Proof of Lemma \ref{lem: div}]
	Let $p_\mu(x) = \mathbb{P}_\mu (x(t) = x)$ be the density function associated with $\mathbb{P}_{\mu}$ and let $p_{\pi \mu}$ be the density function associated with $\mathbb{P}_{\pi \mu}$. We denote by $\pi( \mathcal{C}_1 | h)$ to be the probability of the caching policy choosing cache configuration $\mathcal{C}_1$ given the history $h$. Then,
	\begin{multline*}
		p_{\pi \mu} (C(1), x_1, C(2), x_2, \ldots, C(t_i+1)) \\= \pi(C(1)) p_{\mu}(x_1) \pi(C(2) | C(1), x_1) \cdots \pi(C(t_i+1) | C(1), x_1, \ldots, x_{t_i}).
	\end{multline*}
Thus,
\begin{align*}
	D(\mathbb{P}_{\pi \mu}, \mathbb{P}_{\pi \mu'}) &= \sum_{g \in \mathcal{G}_i} p_{\pi \mu} (g) \log \left( \frac{p_{\pi \mu}(g)}{p_{\pi \mu'}(g)}  \right) \\
	&= \sum_{X^i = (x(1), \ldots, x(t_i))} p_{\mu }(X^i) \log \left( \frac{p_{\mu}(X^i)}{p_{\mu'}(X^i)} \right) \\
	&=  t_i D(\mathbb{P}_{\mu}, \mathbb{P}_{\mu'}),
\end{align*}
where the second last step follows from the fact that the log likelihood does not depend upon the caching policy and the last step follows from the independence of the file requests. 
\end{proof} 
\section{Proof of part (b) of Theorem \ref{thm: restricted_stochastic_main}} \label{sec: app_restricted_stochastic_ub_ftpl}
In this section, we prove an upper bound on the regret of FTPL($\alpha \sqrt{t}$) when cache updates are restricted to $s+1$ fixed time slots under adversarial requests. For any file $k$, at time $t$, let $\alpha_k(t)$ denote the empirical average number of requests received by file $k$. We also define, for any file $k$ after $t$ files have been requested from the cache and for $2 \leq i \leq s$,
\begin{align*}
\hat{\mu}_k(t) &:=\frac{c_k(t) + \eta_t \gamma(k)}{t}, \numberthis \label{eqn: perturbed_average} \\
t_i &:=\sum_{j=1}^{i-1} r_j, \numberthis \label{eqn: cum_sums} \\
\eta^i &:= \eta_{t_i},
\end{align*}
 where $c_k(t)$ denotes the number of requests received by file $k$ at time $t$. Thus,
\begin{align*}
R_{S}^{FTPL(\alpha \sqrt{t})}(T) &=\mathbb{E}\left[\sum_{t=1}^{T} \mathbbm{1}\{x(t) \in \mathcal{C}\}-\mathbbm{1}\{x(t) \in C(t)\}\right]  \\
&\stackrel{(a)}{\leq}  r_1 + \mathbb{E}\left[\sum_{t=r_1+1}^{T} \mathbbm{1}\{x(t) \in \mathcal{C}\}-\mathbbm{1}\{x(t) \in C(t)\}\right]  \\
&=r_1+\mathbb{E}\left[\sum_{i=2}^{s} \sum_{t=t_i+1}^{t_{i+1}} \mathbbm{1}\{x(t) \in \mathcal{C}\}-\mathbbm{1}\{x(t) \in C(t)\}\right]  \\
&= r_1+\sum_{i=2}^{s} \sum_{t=t_i+1}^{t_{i+1}}\mathbb{E}\left[ \mathbbm{1}\{x(t) \in \mathcal{C}\}-\mathbbm{1}\{x(t) \in C(t)\}\right] \\
& \stackrel{(b)}{=}r_1+ \sum_{i=2}^{s} r_i \, \mathbb{E}\left[ \mathbbm{1}\{x(t_i+1) \in \mathcal{C}\}-\mathbbm{1}\{x(t_i+1) \in C(t_i+1)\}\right], \\
\end{align*}
where (a) follows from bounding the regret for the first $r_1$ time slots by $r_1$ and (b) follows from the fact that $C(t)$ remains the same from $t_i+1$ to $t_{i+1}$ for $1 \leq i \leq s$. 
\begin{align*}
R_{S}^{FTPL(\alpha \sqrt{t})}(T) & \stackrel{(c)}{\leq} r_1 +  \sum_{i=2}^{s} r_i \, \mathbb{E}\left[ \left [\left(\sum_{j \in \mathcal{C}} \mu_{j}-\sum_{k \in C\left (t_i +1\right )} \mu_{k}\right)\right] \right ]\\
&\leq r_1 + \sum_{i=2}^{s} r_i  \, \mathbb{E}\left[ \sum_{j=1}^{C} \sum_{k=C+1}^{L} \Delta_{j, k} \mathbbm{1}\left \{j \notin C\left (t_i +1 \right ), k \in C\left (t_i +1\right ) \right \}\right] \\
&\leq r_1 + \sum_{i=2}^{s} r_i  \, \mathbb{E}\left[ \sum_{j=1}^{C} \sum_{k=C+1}^{L} \Delta_{j, k} \mathbbm{1}\left(\hat{\mu}_{k} \left (t_i \right )>\hat{\mu}_{j} \left (t_i \right )\right)\right]\\
& \stackrel{(d)}{\leq} r_1 + \sum_{i=2}^{s} r_i \, \mathbb{E}\Biggl[ \sum_{j=1}^{C} \sum_{k=C+1}^{L} \Delta_{j, k}\biggl(\mathbbm{1}\left\{\hat{\mu}_{j}\left (t_i \right )-\mu_{j} \leq-\Delta_{j, k} / 2\right\} 
\\&+\mathbbm{1}\left\{\hat{\mu}_{k}\left (t_i  \right )-\mu_{k}>\Delta_{j, k} / 2\right\}\biggr)\Biggr],
\end{align*}
where $(c)$ follows from the tower property of conditional expectation (i.e., condition on $C(t)$ inside the outer expectation) and $(d)$ follows from adding $\Delta_{j,k}$ on both sides and using the fact that at least one of $(\hat{\mu}_{k}\left (t_i \right )-\mu_{k} )$ and $(\mu_{j} - \hat{\mu}_{j}\left (t_i \right )) $ must be greater than $\Delta_{j,k}/2$. Using \eqref{eqn: perturbed_average},
\begin{align*}
R_{S}^{FTPL(\alpha \sqrt{t})}(T) &\leq r_1 + \sum_{i=2}^{s} r_i  \, \mathbb{E}\Biggl[\sum_{j=1}^{C} \sum_{k=C+1}^{L} \Delta_{j, k} \Biggl ( \mathbbm{1}\left\{\alpha_j (t_i)  -\mu_{j} \leq-\Delta_{j, k} / 4\right\} \\&+ \mathbbm{1}\left\{\eta_{t_i }\gamma(j)\leq-t_i \Delta_{j, k} / 4\right\} \Biggr ) \Biggr] \\
&+ \sum_{i=2}^{s} r_i \, \mathbb{E}\Biggl[ \sum_{j=1}^{C} \sum_{k=C+1}^{L} \Delta_{j, k} \Biggl ( \mathbbm{1}\left\{\alpha_k(t_i) -\mu_{k}>\Delta_{j, k} / 4\right\}\ \\&+ \mathbbm{1} \left \{ \eta_{t_i }\gamma(k) > t_i \Delta_{j, k} / 4 \right \} \Biggr) \Biggr],  \numberthis \label{eqn: restricted_stoch_ftpl_regret}
\end{align*}
which follows from a union bounding argument similar to (d). Using Hoeffding's inequality (\cite{hoeffding}), we obtain
\begin{align*}
\mathbb{P}\left(\alpha_{j}(t)-\mu_{j} \leq-\Delta_{j, k} / 4\right) &\leq e^{-t \Delta_{j, k}^{2} / 8}, \\
\mathbb{P}\left(\alpha_{k}(t)-\mu_{k}>\Delta_{j, k} / 4\right) &\leq e^{-t^2 \Delta_{j, k}^{2} / 8}.
\end{align*}
Also,
\begin{align*}
\mathbb{P}\left(\frac{\eta_t}{t}\gamma(j) \leq -\Delta_{j, k} / 4\right) = \mathbb{P}\left(\frac{\eta_t}{t}\gamma(k) > \Delta_{j, k} / 4\right) \leq e^{-t^2 \Delta_{j,k}^2/32\eta_t^2}.
\end{align*}
By taking the expectation inside the summation in \eqref{eqn: restricted_stoch_ftpl_regret}, we obtain the following upper bound for the regret:
\begin{align*}
R_{S}^{FTPL(\alpha \sqrt{t})}(T) &\leq r_1+ 2\sum_{j=1}^{C} \sum_{k=C+1}^{L} \sum_{i=2}^{s} r_i \, \Delta_{j, k} \left ( e^{-t_i \Delta_{j, k}^{2} / 8} + e^{-t_i \Delta_{j,k}^2/32 \alpha^2} \right ) \\
& = r_1+ 2\sum_{j=1}^{C} \sum_{k=C+1}^{L} \sum_{i=2}^{s} r_i \, \Delta_{j, k} \left ( e^{-\left ( \sum_{j=1}^{i-1}r_j \right )   \Delta_{j, k}^{2} / 8} + e^{-\left ( \sum_{j=1}^{i-1}r_j \right )  \Delta_{j,k}^2/32 \alpha^2} \right ) \numberthis \label{eqn: restricted_stochastic_ub_r_i_general}\\
& \leq r_1+ 2 \underset{2 \leq a \leq s}{\max} r_a \sum_{j=1}^{C} \sum_{k=C+1}^{L} \sum_{i=2}^{s}  \, \Delta_{j, k} \left ( e^{-\left ( (i-1) \underset{1 \leq a \leq s}{\min} r_a \right )   \Delta_{j, k}^{2} / 8} + e^{-(i-1) \underset{1 \leq a \leq s}{\min} r_a \Delta_{j,k}^2/32 \alpha^2} \right ) \\
& \leq r_1 + 2 \frac{\underset{2 \leq a \leq s}{\max} r_a}{\underset{1 \leq a \leq s}{\min} r_a}\sum_{j=1}^{C} \sum_{k=C+1}^{L} \frac{8}{\Delta_{j,k}}+ \frac{32 \alpha^2}{\Delta_{j,k}} \\
& \leq r_1 + 2C(L-C)  \frac{\underset{2 \leq a \leq s}{\max} r_a}{\underset{1 \leq a \leq s}{\min} r_a} \left ( \frac{8}{\Delta_{\min}}+ \frac{32 \alpha^2}{\Delta_{\min}}  \right ). 
\end{align*}
\section{Proof of Theorem \ref{thm: constrained_stochastic_main}} \label{sec: app_constrained_stochastic_ub_ftpl}
In this section, we prove an upper bound on the regret of FTPL$(\alpha \sqrt{t})$ when cache updates are restricted to periodic time slots. Recall that $t_0'=\max \left \{r,\frac{8}{\Delta_{\min}^2} \log \left ( {L^2}\right )  , \frac{32 \alpha^2}{\Delta_{\min}^2 }\log \left ( {L^2}\right ) \right \}$ and we assume that $L \geq 3$. We bound the regret for the first $t_0'$ rounds by $t_0'$. This technique of bounding the regret of an initial period by its worst-case regret and then using normal methods to bound the regret for $t>t_0'$ is based on \cite{hedge}.
Thus, we get the following expression for the regret: 
\begin{align*}
     R_{S}^{\textrm{FTPL}(\alpha \sqrt{t}) } &\leq \lceil t_0' \rceil +r \mathbb{E}\left[\sum_{t= \left \lceil \frac{t_0'}{r} \right \rceil +1}^{T/r} \mathbbm{1}\{x(rt) \in \mathcal{C}\}-\mathbbm{1}\{x(rt) \in C(rt)\}\right] \\
     &\leq 1+t_0' +  2r\sum_{t=\left \lceil \frac{t_0'}{r} \right \rceil+1}^{T/r} \sum_{j=1}^{C} \sum_{k=C+1}^{L}  \Delta_{j,k} \left(e^{-rt \Delta_{j, k}^{2} / 8}+e^{-rt \Delta_{j, k}^{2} / 32 \alpha^{2}}\right),
\end{align*}
which follows from \eqref{eqn: restricted_stochastic_ub_r_i_general}. The function $f(u)=$ $u e^{-u^{2} / 2}$ is decreasing on $[1,+\infty)$. Since $\Delta_{j,k} \geq \Delta_{\min}, j \in \mathcal{C}, k \notin \mathcal{C}$, we get
 \begin{align*}
\Delta_{j,k} e^{-rt \Delta_{j,k}^{2} / 8} &=\frac{2}{\sqrt{rt}} f\left(\frac{\sqrt{rt} \Delta_{j,k}}{2}\right) \\&\leq \frac{2}{\sqrt{rt}} f\left(\frac{\sqrt{rt} \Delta_{\min}}{2}\right)=\Delta_{\min} e^{-t \Delta_{\min}^{2} / 8} .
  \end{align*}
Note that for $t \geq \left \lceil \frac{t_0'}{r} \right \rceil + 1$, $t> \frac{4}{r\Delta_{\min}^2}$ holds as $t_0' \geq \frac{8}{\Delta_{\min}^2}$. Thus,
\begin{multline*}
        2r\sum_{t=\left \lceil \frac{t_0'}{r} \right \rceil+1}^{T/r} \sum_{j=1}^{C} \sum_{k=C+1}^{L}  \Delta_{j,k} \left(e^{-rt \Delta_{j, k}^{2} / 8}+e^{-rt \Delta_{j, k}^{2} / 32 \alpha^{2}}\right) \\ \leq 2rL^2\Delta_{\min} \sum_{t=\left \lceil \frac{t_0'}{r} \right \rceil+1}^{T/r}   e^{-rt \Delta_{\min}^{2} / 8}+e^{-rt \Delta_{\min}^{2} / 32 \alpha^{2}}.
\end{multline*}
Now,
\begin{align*}
    L^2\sum_{t=\left \lceil \frac{t_0'}{r} \right \rceil+1}^{T/r}   e^{-rt \Delta_{\min}^{2} / 8}+e^{-rt \Delta_{\min}^{2} / 32 \alpha^{2}}
    &\leq  L^2 \sum_{t=\left \lceil \frac{t_0'}{r} \right \rceil+1}^{T/r} \biggl \{ e^{-r \left (\left \lceil \frac{t_0'}{r} \right \rceil+t-\left \lceil \frac{t_0'}{r} \right \rceil \right) \Delta_{\min}^2/8} \\ &+  e^{-r \left (\left \lceil \frac{t_0'}{r} \right \rceil+ t-\left \lceil \frac{t_0'}{r} \right \rceil \right )\Delta_{\min}^{2} / 32 \alpha^{2}} \biggr \}\\ 
 &\leq \sum_{t=1}^{T/r-\left \lceil \frac{t_0'}{r} \right \rceil} e^{- \frac{rt \Delta_{\min}^2}{8} } + e^{- \frac{rt \Delta_{\min}^{2} }{ 32 \alpha^{2}} } \\&\leq  \frac{1}{r} \left ( \frac{8}{\Delta_{\min}^2 } + \frac{32 \alpha^2}{\Delta_{\min}^2 } \right ). 
\end{align*}
This gives the following upper bound on the regret:
\begin{align*}
   R_{S}^{\textrm{FTPL}(\alpha \sqrt{t}) }  \leq 1+t_0' + 2 \left ( \frac{8}{\Delta_{\min}} + \frac{32 \alpha^2}{\Delta_{\min}} \right ). 
\end{align*}    

\section{Proof of part (a) of Theorem \ref{thm: restricted_adversarial_main}} \label{sec: app_restricted_adversarial_lb} 
In this section, we prove a lower bound on the regret of any policy $\pi$ when cache updates are restricted under adversarial file requests. The proof uses a technique called the probabilistic method from \cite{sigmetrics} which was pioneered by Erd\H os \cite{erdos}. Consider a request sequence $\{x_t\}_{t=1}^{T}$ with a joint probability distribution defined on it. Then, we have the following lower bound on the regret incurred by any policy $\pi$:
\begin{equation}
R_{A}^{\pi}(T) \geq \mathbb{E}_{\{x_t\}_{t=1}^{T}} R_{A}^{\pi}(\{x_t\}_{t=1}^{T},T). \label{eqn: erdos_lb}
\end{equation}
Thus, the proof proceeds by defining a request sequence first and then using \eqref{eqn: erdos_lb} to lower bound the regret. Throughout this proof, if it is not explicitly mentioned with respect to what the expectation is being taken, it can be assumed that the expectation is being taken with respect to $\{x_t\}_{t=1}^{T}$.
\subsection{Lower Bound when $\mathbf{C=1}$}
We consider an adversarial request sequence where files are requested from the first $2$ files uniformly, and the same file is requested $r_i$ times from $t_i=\sum_{j=0}^{i-1}r_j$ to $t_i+r_i-1$, where $r_0=1$ and $1 \leq i \leq s$. To be precise, for each $1 \leq i \leq s$, a file is drawn uniformly at random from the first two files at $t=\sum_{j=0}^{i-1}r_j$ and is repeatedly requested till $t=t_i+r_i-1$, i.e., a total of $r_i$ times. Throughout this proof, ``phase $i$" refers to the time period from $t=\sum_{j=0}^{i-1}r_j$ to $t=t_i+r_i-1$ (both points included). Let $W_i$ be a Bernoulli random variable indicating whether file 1 was requested in the $i^{\textrm{th}}$ phase or not. Note that $\{W_i\}_{i=1}^{s}$ are i.i.d. Thus, the reward obtained by the optimal offline policy corresponds to:
\begin{align*}
 \underset{y \in \mathcal{Y}}{\max} \left \langle \boldsymbol{y}, \boldsymbol{X}_{T+1} \right \rangle &=  \underset{y \in \mathcal{Y}}{\max} \left \langle \boldsymbol{y}, \sum_{t=1}^{T} \boldsymbol{x}_{t} \right \rangle \\
 &=  \underset{y \in \mathcal{Y}}{\max} \left \langle \boldsymbol{y}, \sum_{i=1}^{s} \sum_{t=t_i}^{t_i+r_i-1} \boldsymbol{x}_{t} \right \rangle \\
 &=  \underset{y \in \mathcal{Y}}{\max} \left \langle \boldsymbol{y}, \sum_{i=1}^{s} r_i \, \boldsymbol{x}_{t_i} \right \rangle,
 \end{align*}
as the specific request sequence chosen here is constant in each phase. As $C=1$, $\boldsymbol{y}$ is a one-hot encoded vector and thus, we have that
 \begin{align*}
  \underset{y \in \mathcal{Y}}{\max} \left \langle \boldsymbol{y}, \boldsymbol{X}_{T+1} \right \rangle &= \max \left \{ \sum_{i=1}^{s} r_i W_i, T-\sum_{i=1}^{s} r_i W_i \right \} = \frac{T}{2} + \left |\frac{T}{2} - \sum_{i=1}^{s} r_i W_i\right |.
\end{align*}
The reward of any algorithm can be upper bounded in the following way:
\begin{align*}
\sum_{t=1}^{T} \left \langle \boldsymbol{y}_t, \boldsymbol{x}_{t} \right \rangle &= \sum_{i=1}^{s} \sum_{t=t_i}^{t_i+r_i-1} \left \langle \boldsymbol{y}_t, \boldsymbol{x}_{t} \right \rangle \\
&= \sum_{i=1}^{s} r_i \left \langle \boldsymbol{y}_{t_i}, \boldsymbol{x}_{t_i} \right \rangle,
\end{align*}
as the specific request sequence chosen here and the cache configuration is constant in each phase. Thus, the expected reward of any algorithm can be upper bounded in the following way, where the expectation is taken with respect to the request sequence only:
\begin{align*}
\mathbb{E}_{\{x_t\}_{t=1}^{T}} \left [ \sum_{t=1}^{T} \left \langle \boldsymbol{y}_t, \boldsymbol{x}_{t} \right \rangle \right ] &= \mathbb{E}_{\{x_t\}_{t=1}^{T}} \left [ \sum_{i=1}^{s} r_i \left \langle \boldsymbol{y}_{t_i}, \boldsymbol{x}_{t_i} \right \rangle \right ]\\
&  \sum_{i=1}^{s} r_i  \mathbb{E}_{\{x_t\}_{t=1}^{T}} \left [\left \langle \boldsymbol{y}_{t_i}, \boldsymbol{x}_{t_i} \right \rangle  \right ].
\end{align*}
In each time slot, the cache update happens before the file request arrives and hence $\boldsymbol{y}_t$ is independent of $\boldsymbol{x}_t$, $1 \leq t \leq T$. Thus, we have
\begin{align*}
\mathbb{E}_{\{x_t\}_{t=1}^{T}} \left [ \sum_{t=1}^{T} \left \langle \boldsymbol{y}_t, \boldsymbol{x}_{t} \right \rangle \right ] &=   \sum_{i=1}^{s} r_i  \mathbb{E}_{\{x_t\}_{t=1}^{T}} \left [\left \langle \boldsymbol{y}_{t_i}, \boldsymbol{x}_{t_i} \right \rangle  \right ] \\
&=   \sum_{i=1}^{s} r_i  \left \langle \boldsymbol{y}_{t_i}, \mathbb{E}_{\{x_t\}_{t=1}^{T}} \left [  \boldsymbol{x}_{t_i} \right ] \right \rangle  \\
&=   \frac{1}{2C} \sum_{i=1}^{s} r_i  \left (\boldsymbol{y}_{t_i}(1)+\boldsymbol{y}_{t_i}(2) \right ) \\
&\leq \frac{1}{2} \sum_{i=1}^{s} r_i = \frac{T}{2}. \numberthis \label{eqn: algo_reward_ub}
\end{align*} 
Thus,
\begin{align*}
R_A^{\pi}(T) \geq  \mathbb{E}_{\{x_t\}_{t=1}^{T}} \left [ \left |\frac{T}{2} - \sum_{i=1}^{s} r_i W_i\right | \right ].
\end{align*}
Now, we lower bound $\mathbb{E}[|M|]$, where $M=\frac{T}{2} - \sum_{i=1}^{s} r_i W_i= \sum_{i=1}^{s} r_i\left ( W_i - \frac{1}{2}\right )$, as $\sum_{i=1}^{s}r_i=T$. We denote $m_i=r_i\left ( W_i - \frac{1}{2}\right )$. Thus, $M=\sum_{i=1}^{s} m_i$, where $\mathbb{E}[m_i]=0$ and $\sigma_{i}^2=\mathbb{E}[m_i^2]=\frac{r_i^2}{4}$. Note that $\sigma_M=\sqrt{\sum_{i=1}^{s} \sigma_i^2}=\frac{1}{2}\sqrt{\sum_{i=1}^{s}r_i^2}$. Now, using the Markov inequality, 
\begin{align*}
    \mathbb{E}[|M|] \geq \sigma_M \mathbb{P} (|M| \geq \sigma_M) &\geq \sigma_M \mathbb{P} (\frac{M}{\sigma_M} \geq 1) \\
    &\geq \sigma_M \mathbb{P} (\frac{\sum_{i=1}^{s}m_i}{\sigma_M} \geq 1).
\end{align*}
Using the Berry-Esseen theorem, we have that
\begin{align*}
    \left|\operatorname{Pr}\left(\frac{\sum_{i=1}^{s}m_i}{\sigma_M} \leq 1 \right)-\Phi(1)\right| \leq \frac{C_0}{\sigma_M} \max _{1 \leq i \leq n} \frac{\rho_{i}}{\sigma_{i}^{2}} ,
\end{align*}
where $\rho_i= \mathbb{E}[|m_i^3|]=\frac{r_i^3}{8}$, $C_0$ is a constant and $\Phi(\cdot)$ is the CDF of the standard Gaussian random variable. Thus, we get
\begin{align*}
    \operatorname{Pr}\left(\frac{\sum_{i=1}^{s}m_i}{\sigma_M} \geq 1 \right) &\geq 1 - \Phi(1) - C_0 \frac{\underset{1 \leq i \leq s}{\max} r_i }{2\sigma_M}   \\
    &\geq 0.15 - C_0 \frac{\underset{1 \leq i \leq s}{\max} r_i }{\sqrt{\sum_{i=1}^{s}r_i^2}}. 
\end{align*}
We thus have:
\begin{align*}
 R_A^{\pi}(T) &\geq  \mathbb{E}_{\{x_t\}_{t=1}^{T}} \left [   \max \left \{ \sum_{i=1}^{s} r_i W_i, T-\sum_{i=1}^{s} r_i W_i \right \}  \right ]\\
 &\geq \frac{T}{2} +  \frac{1}{2}\sqrt{\sum_{i=1}^{s}r_i^2} \left ( 0.15 - C_0 \frac{\underset{1 \leq i \leq s}{\max} r_i }{\sqrt{\sum_{i=1}^{s}r_i^2}} \right ). 
\end{align*}
\subsection{Lower bound for general $L,C$}
In this section, we extend the result in the previous section for a general $L,C$ value. We consider an analogous adversarial request sequence where, files are requested from the first $2C$ files uniformly, and the same file is requested $r_i$ times from $t_i=\sum_{j=0}^{i-1}r_j$ to $t_i+r_i-1$, where $r_0=1$ and $1 \leq i \leq s$. To be precise, for each $1 \leq i \leq s$, a file is drawn uniformly at random from the first two files at $t=\sum_{j=0}^{i-1}r_j$ and is repeatedly requested till $t=t_i+r_i-1$, i.e., a total of $r_i$ times. We denote by $W_i$ the file requested in the $i^{\text{th}}$ phase. 

 We use the balls-into-bins technique from the proof of Lemma 1 of \cite{sigmetrics}. A bin is associated with each file from $1,\ldots,2C$ where a request for that file is equivalent to a ball being thrown into that bin. The bins are numbered as $1,2, \ldots, 2 C$ and every two consecutive bins $\{(2 i-1,2 i)\}, 1 \leq i \leq C$ are combined to form $C$ Super bins. Denote by $Z_{i}, i=1,2, \ldots, C$ the number of balls in the $i^{\text {th }}$ super bin. Conditioned on $Z_{i}$, the number of balls in the bins $2 i-1$ and $2 i$ are jointly distributed as $\left(V, Z_{i}-V\right)$, where $V$ is a binomial random variable with parameter $\left(Z_{i}, \frac{1}{2}\right)$. Let $H_{i}$ denote the number of balls in the bin containing the maximum number of balls among bins $2 i-1$ and $2 i$. Then, as shown in the previous section, when $\forall 1 \leq i \leq C, Z_{i}>0$ :
\begin{align*}
\mathbb{E}\left(H_{i} \mid Z_{i}\right) \geq \frac{Z_i}{2} + \frac{1}{2}\sqrt{\sum_{j=1}^{s}r_j^2 \mathbb{I}_{W_j \in \{(2 i-1,2 i)\} }} \left ( 0.15 - C_0 \frac{\underset{1 \leq j \leq s, W_j \in \{(2 i-1,2 i)\} }{\max} r_j }{\sqrt{\sum_{j=1}^{s}r_j^2 \mathbb{I}_{W_j \in \{(2 i-1,2 i)\} } }  } \right ),
 \end{align*}
The expected number of requests obtained by the best offline policy can be lower bounded as:
\begin{align*}
\underset{y \in \mathcal{Y}}{\max} \left \langle \boldsymbol{y}, \boldsymbol{X}_{T+1} \right \rangle    &\geq \mathbb{E} \left [ \sum_{i=1}^{C} H_{i} \right ]  \\
    &= \sum_{i=1}^{C} \mathbb{E} \left [  H_{i} \right ] \\
    &\stackrel{(a)}{=} C H_1 \\
    &\stackrel{(b)}{=} C \mathbb{E} \left [ \mathbb{E}\left[H_{1} \mathbb{I}\left(Z_{1}>0\right) \mid Z_{1}\right] \right ],
    \end{align*}
    where (a) follows from the request sequence being symmetric across the $2C$ files and (b) follows from the tower property of conditional expectation as shown in \cite{sigmetrics}. Thus, 
    \begin{align*}
 \underset{y \in \mathcal{Y}}{\max} \left \langle \boldsymbol{y}, \boldsymbol{X}_{T+1} \right \rangle   &\geq \frac{C \mathbb{E}[Z_1]}{2} + C \mathbb{E} \left [ \frac{1}{2}\sqrt{\sum_{i=1}^{s}r_i^2 \mathbb{I}_{W_i \in \{(1,2)\} }} \left ( 0.15 - C_0 \frac{\underset{1 \leq i \leq s, W_i \in \{(1,2)\} }{\max} r_i }{\sqrt{\sum_{i=1}^{s}r_i^2 \mathbb{I}_{W_i \in \{(1,2)\} } }  } \right ) \right ]\\ 
    &\geq \frac{T}{2} +  \frac{C}{2} \mathbb{E} \left [  0.15 \sqrt{\sum_{i=1}^{s}r_i^2 \mathbb{I}_{W_i \in \{(1,2)\} }} - C_0 \underset{1 \leq i \leq s }{\max} r_i  \right ], \numberthis \label{eqn: res_ad_lb_final}
\end{align*}
as $\mathbb{E}[Z_1]=T/C$. Using (20) of \cite{sigmetrics}, 
\begin{align*}
\mathbb{E} \left [ \sqrt{ \sum_{i=1}^{s}r_i^2 \mathbb{I}_{W_i \in \{(1,2)\} } }\right ] &\geq \sqrt{\mathbb{E}\left( \sum_{i=1}^{s}r_i^2 \mathbb{I}_{W_i \in \{(1,2)\} } \right)}\left(1-\frac{\operatorname{Var}\left(\sum_{i=1}^{s}r_i^2 \mathbb{I}_{W_i \in \{(1,2)\} } \right)}{2\left(\mathbb{E}\left(\sum_{i=1}^{s}r_i^2 \mathbb{I}_{W_i \in \{(1,2)\} }\right)\right)^{2}}\right) \\
&= \sqrt{ \frac{\sum_{i=1}^{s}r_i^2}{C} } \left(1-\frac{(C-1)\left(\sum_{i=1}^{s}r_i^4 \right)}{2\left(\sum_{i=1}^{s}r_i^2 \right)^{2}}\right),
\end{align*}
as $\operatorname{Var} \left (\mathbb{I}_{W_i \in \{(1,2)\} } \right )=\frac{1}{C}\left(1-\frac{1}{C}\right)$ and $\mathbb{E} \left (\mathbb{I}_{W_i \in \{(1,2)\} } \right )=\frac{1}{C}$. Similar to \eqref{eqn: algo_reward_ub},  
\begin{align*}
\mathbb{E}_{\{x_t\}_{t=1}^{T}} \left [ \sum_{t=1}^{T} \left \langle \boldsymbol{y}_t, \boldsymbol{x}_{t} \right \rangle \right ] &=  \frac{1}{2C} \sum_{i=1}^{s} r_i  \left (\boldsymbol{y}_{t_i}(1)+\cdots+\boldsymbol{y}_{t_i}(2C) \right ) \\
&\leq \frac{1}{2} \sum_{i=1}^{s} r_i = \frac{T}{2}. 
\end{align*} 
Combining the above result and \ref{eqn: res_ad_lb_final} gives us the following lower bound on the regret:
\begin{align*}
 R^{\pi}_{A}(T)   \geq \frac{1}{2}   \left(0.15 \sqrt{ C\sum_{i=1}^{s}r_i^2 } \left(1-\frac{(C-1)\left(\sum_{i=1}^{s}r_i^4 \right)}{2\left(\sum_{i=1}^{s}r_i^2 \right)^{2}}\right) -0.6 \, C  \underset{1 \leq i \leq s}{\max} r_i  \right). 
\end{align*}
\section{Proof of part (b) of Theorem \ref{thm: restricted_adversarial_main}} \label{sec: app_restricted_adversarial_ub_ftpl}
In this section, we prove an upper bound on the regret of FTPL($\alpha \sqrt{t}$) when cache updates are restricted to after every $r$ slots only under adversarial file requests. The proof is based on the proof of Proposition 4.2 of \cite{sqrt_t}. Define the following potential function at each time slot $t$:
\begin{equation*}
    \Phi_{t}(\boldsymbol{x})=\mathbb{E}_{\boldsymbol{\gamma} \sim \mathcal{N}(0, I)}\left[\max _{\boldsymbol{y} \in \mathcal{Y}}\langle\boldsymbol{y}, \boldsymbol{x}+\eta \boldsymbol{\gamma}\rangle\right]
\end{equation*}

The regret incurred in this setting can be expressed in the following way:
\begin{align*}
R_{A}^{\textrm{FTPL}(\alpha \sqrt{t})}(T) &= \max _{\boldsymbol{y} \in \mathcal{Y}}\left\langle\boldsymbol{y}, \boldsymbol{X}_{T+1}\right\rangle - \sum_{t=1}^{T} \mathbb{E}_{\boldsymbol{\gamma}}\left[\left\langle\boldsymbol{y}_{t}, \boldsymbol{x}_{t}\right\rangle\right]  \\
&= \max _{\boldsymbol{y} \in \mathcal{Y}}\left\langle\boldsymbol{y}, \boldsymbol{X}_{T+1}\right\rangle - \sum_{i=1}^{s} \mathbb{E}_{\boldsymbol{\gamma}}\left[\left\langle\boldsymbol{y}_{ \sum_{j=0}^{i-1}r_j } , \sum_{t=\sum^{{i-1}}_{k=0} r_k }^{\sum^{{i}}_{k=1} r_k} \boldsymbol{x}_{t }\right\rangle\right],  \\  \numberthis \label{eqn: ftpl_adversarial_regret_restricted} \\
\end{align*}
where we define $r_0=1$. This follows from the fact that the cache configuration can change only at the pre-defined $s$ fixed time slots. Thus, we essentially have a time horizon of $s$, but each time slot $i, 1 \leq i \leq s$ contains $r_i$ requests instead. For brevity, we define 
\begin{align*}
\boldsymbol{y}^{i} := \boldsymbol{y}_{ \sum_{j=0}^{i-1}r_j },\quad \boldsymbol{x}^{i} :=  \sum_{t=\sum^{{i-1}}_{k=0} r_k }^{\sum^{{i}}_{k=1} r_k} \boldsymbol{x}_{t },\quad \eta^i := \eta_{\sum_{j=0}^{i-1}r_j },\quad \boldsymbol{X}^{i} := \sum_{j=1}^{i} \boldsymbol{x}^{j},\quad 1 \leq i \leq s.
\end{align*}

We define the potential function $\Phi_{i} : \mathbb{R}^{L} \rightarrow \mathbb{R}$ for $1 \leq i \leq s$ in the following way:
\begin{align*}
\Phi_{i}(\boldsymbol{x})=\mathbb{E}_{\boldsymbol{\gamma}}\left[\max _{\boldsymbol{y} \in \mathcal{Y}}\left\langle\boldsymbol{y}, \boldsymbol{x}+\eta^{i} \boldsymbol{\gamma}\right\rangle\right],
\end{align*}
 where $\mathcal{Y}$ is the set of possible cache configurations, i.e., the set $\{y \in \{0,1\}^{L}: \|y\|_1 \leq C \}$. As shown in the proof of Proposition 4.1 in \cite{sqrt_t}, 
 \begin{align*}
 \nabla \Phi_{i}\left(\boldsymbol{X}^{i}\right) &=\left.\nabla \mathbb{E}_{\boldsymbol{\gamma}}\left[\left\langle\boldsymbol{y}^{i}, \boldsymbol{x}+\eta^{i} \gamma\right\rangle\right]\right|_{\boldsymbol{x}=\boldsymbol{X}^{i}}=\mathbb{E}_{\boldsymbol{\gamma}}\left[\boldsymbol{y}^{i}\right]. \\
\mathbb{E}_{\boldsymbol{\gamma}}\left[\left\langle\boldsymbol{y}^{i}, \boldsymbol{x}^{i}\right\rangle\right]  &= \left\langle\nabla \Phi_{i}\left(\boldsymbol{X}^{i}\right), \boldsymbol{X}^{i}-\boldsymbol{X}^{i-1}\right\rangle\\
&=\Phi_{i}\left(\boldsymbol{X}^{i}\right)-\Phi_{i}\left(\boldsymbol{X}^{i-1}\right)-\frac{1}{2}\left\langle\boldsymbol{x}^{i}, \nabla^{2} \Phi_{i}\left(\widetilde{\boldsymbol{X}}^{i}\right) \boldsymbol{x}^{i} \right\rangle,
 \end{align*}
where $\widetilde{\boldsymbol{X}}^{i}=\boldsymbol{X}^{i-1}+\theta^{i} \boldsymbol{x}^{i}$, for some $\theta^{i} \in[0,1]$ using Taylor's theorem. Thus, we have
\begin{align*}
&\sum_{i=1}^{s} \mathbb{E}_{\boldsymbol{\gamma}}\left[\left\langle\boldsymbol{y}^{i}, \boldsymbol{x}^{i}\right\rangle\right] \\
&=\Phi_{s}\left(\boldsymbol{X}^{s}\right) -\Phi_{1}\left(\boldsymbol{X}^{0}\right)+\sum_{i=2}^{s}\left[\Phi_{i-1}\left(\boldsymbol{X}^{i-1}\right)- \Phi_{i}\left(\boldsymbol{X}^{i-1}\right)\right] \\
&-\frac{1}{2} \sum_{i=1}^{s}\left\langle\boldsymbol{x}^{i}, \nabla^{2} \Phi_{i}\left(\widetilde{\boldsymbol{X}}^{i}\right) \boldsymbol{x}^{i}\right\rangle.
\end{align*}
Using Jensen's inequality, we have that
\begin{align*}
\Phi_{s}\left(\boldsymbol{X}_{s}\right ) &= \mathbb{E}_{\boldsymbol{\gamma}}\left[\max _{\boldsymbol{y} \in \mathcal{Y}}\left\langle\boldsymbol{y}, \boldsymbol{X^{s}}+\eta^{s} \boldsymbol{\gamma}\right\rangle\right] \\
&\geq \max _{\boldsymbol{y} \in \mathcal{Y}} \mathbb{E}_{\boldsymbol{\gamma}}\left[\left\langle\boldsymbol{y}, \boldsymbol{X^{s}}+\eta^{s} \boldsymbol{\gamma}\right\rangle\right] \\
&=  \max _{\boldsymbol{y} \in \mathcal{Y}} \left\langle\boldsymbol{y}, \boldsymbol{X^{s}} \right\rangle =  \max _{\boldsymbol{y} \in \mathcal{Y}} \left\langle\boldsymbol{y}, \boldsymbol{X_{T+1}} \right\rangle. 
\end{align*}
Substituting the above results in \eqref{eqn: ftpl_adversarial_regret_restricted}, we get
\begin{align*}
R_{A}^{\textrm{FTPL}(\alpha \sqrt{t})}(T) &\leq \Phi_{1}\left(\boldsymbol{X}^{0}\right)+\sum_{i=2}^{s}\left[ \Phi_{i}\left(\boldsymbol{X}^{i-1}\right) - \Phi_{i-1}\left(\boldsymbol{X}^{i-1}\right) \right] \\
&+ \frac{1}{2} \sum_{i=1}^{s}\left\langle\boldsymbol{x}^{i}, \nabla^{2} \Phi_{i}\left(\widetilde{\boldsymbol{X}}^{i}\right) \boldsymbol{x}^{i}\right\rangle.
\end{align*}
Since $\boldsymbol{x}^{i}$ contains $r_i$ file requests, the quadratic form above may be upper bounded in the following way:
\begin{align*}
\left\langle\boldsymbol{x}^{i}, \nabla^{2} \Phi_{i}\left(\widetilde{\boldsymbol{X}}^{i}\right) \boldsymbol{x}^{i}\right\rangle &\leq r_i \max _{k,j, \boldsymbol{x}}\left(\left|\nabla^{2} \Phi_{i}(\boldsymbol{x})\right|\right)_{k j}  \left\langle\boldsymbol{x}^{i}, \mathbf{1} \right\rangle \\
&= r_i^2 \max _{k,j, \boldsymbol{x}}\left(\left|\nabla^{2} \Phi_{i}(\boldsymbol{x})\right|\right)_{k j},
\end{align*}
where $\mathbf{1}$ is the all ones vector. Moreover, from \cite{sigmetrics},
\begin{align*}
\left(\left|\nabla^{2} \Phi_{i}(\boldsymbol{x})\right|\right)_{kj} \leq \frac{1}{\eta^{i}} \sqrt{\frac{2}{\pi}}.
\end{align*}
Thus,
\begin{align*}
\frac{1}{2} \sum_{i=1}^{s}\left\langle\boldsymbol{x}^{i}, \nabla^{2} \Phi_{i}\left(\widetilde{\boldsymbol{X}}^{i}\right) \boldsymbol{x}^{i}\right\rangle \leq \sqrt{\frac{2}{\pi}} \sum_{i=1}^{s} \frac{r_i^2}{\eta^{i}}.
\end{align*}
It can also be shown that:
\begin{align*}
\Phi_{i}\left(\boldsymbol{X}^{i-1}\right)-\Phi_{i-1}\left(\boldsymbol{X}^{i-1}\right) &= \mathbb{E}_{\boldsymbol{\gamma}}\left[\max _{\boldsymbol{y} \in \mathcal{Y}}\left\langle\boldsymbol{y}, \boldsymbol{X}^{i-1}+\eta^{i} \boldsymbol{\gamma}\right\rangle\right] - \mathbb{E}_{\boldsymbol{\gamma}}\left[\max _{\boldsymbol{y} \in \mathcal{Y}}\left\langle\boldsymbol{y}, \boldsymbol{X}^{i-1}+\eta^{i-1} \boldsymbol{\gamma}\right\rangle\right]\\
&\leq \mathbb{E}_{\boldsymbol{\gamma}}\left[\max _{\boldsymbol{y} \in \mathcal{Y}}\left\langle\boldsymbol{y}, |\eta^{i} - \eta^{i-1}| \boldsymbol{\gamma}\right\rangle\right] \\
&= \left|\eta^{i}-\eta^{i-1}\right| \mathbb{E}_{\boldsymbol{\gamma}}\left[\max _{\boldsymbol{y} \in \mathcal{Y}}\langle\boldsymbol{y}, \gamma\rangle\right].
\end{align*}
We also have from \cite{cohenhazan15} that: 
\begin{align*}
\Phi_{1}\left(\boldsymbol{X}^{0}\right)= \eta^1 \mathbb{E}_{\boldsymbol{\gamma}}\left[\max _{\boldsymbol{y} \in \mathcal{Y}}\left\langle \boldsymbol{y}, \boldsymbol{\gamma}\right\rangle\right ] \leq \eta^1 C \sqrt{2  \log \left ( \frac{N}{C}\right )}.
\end{align*}
Thus, combining the above bounds, we get
\begin{align*}
R_{A}^{\textrm{FTPL}(\alpha \sqrt{t})}(T) &\leq  \eta^1 C \sqrt{2  \log \left ( \frac{N}{C}\right )} + \eta^s C \sqrt{2  \log \left ( \frac{N}{C}\right )}+ \sqrt{\frac{2}{\pi}} \sum_{i=1}^{s} \frac{r_i^2}{\eta^{i}} \\
&\leq  \alpha C \sqrt{2  \log \left ( \frac{N}{C}\right )} + \alpha C \sqrt{T} \sqrt{2  \log \left ( \frac{N}{C}\right )}+ \sqrt{\frac{2}{\pi}} \sum_{i=1}^{s} \frac{r_i^2}{\alpha \sqrt{\sum_{j=0}^{i-1}r_j }}. \\
\end{align*}

\section{Proof of part (a) of Theorem \ref{thm: constrained_adversarial_main} } \label{sec: app_constrained_lb_adversarial}
In this section, we prove a lower bound on the regret of any policy $\pi$ in the restricted switching case where the cache is allowed to update its contents only after every $r$ requests. The proof uses \eqref{eqn: erdos_lb} to bound the regret. We consider the adversarial request sequence where files are requested from the top $2C$ files, and the same file is requested $r$ times in each of the $T/r$ `phases' (using terminology defined in Appendix \ref{sec: app_restricted_adversarial_lb}). To be precise, at the beginning of each phase $i$ where $1 \leq i \leq T/r$, i.e., at $t=1+(i-1)r$, a file is drawn uniformly at random from the first $2C$ files. This file is repeatedly requested $r$ times till $t=ir$. Throughout this proof, if it is not explicitly mentioned with respect to what the expectation is being taken, it can be assumed that the expectation is being taken with respect to $\{x_t\}_{t=1}^{T}$. The reward obtained by the optimal static configuration in hindsight can be bounded in the following way:
\begin{align*}
    \mathbb{E}\left( \max_{y \in \mathcal{Y}}\left \langle \boldsymbol{y}, \sum_{t=1}^{T}\boldsymbol{x}_{t} \right \rangle \right) &= r\mathbb{E}\left( \max_{y \in \mathcal{Y}} \left \langle \boldsymbol{y}, \sum_{t=1}^{T/r}\boldsymbol{x}_{t} \right \rangle \right)\\
    & \geq r\left(  \frac{T}{2r}+\sqrt{\frac{C T}{2r \pi}}-\frac{\sqrt{r}(\sqrt{2}+1) C^{3 / 2}}{2 \sqrt{2 \pi T}}-\sqrt{\frac{2}{\pi}} \frac{rC^{2}}{T}\right)\\
    & = \frac{T}{2}+\sqrt{\frac{CrT}{2\pi }}-\Theta\left({\frac{r\sqrt{r}}{\sqrt{T}}}\right),
\end{align*}
where in the second last step, we used Theorem 2 of \cite{sigmetrics}. For $1 \leq i \leq s$, let 
\begin{align*}
\boldsymbol{y}^{i} :=\sum_{t=1+(i-1)r}^{ir} \boldsymbol{y}_t,
\end{align*}
which is the sum of the caching configuration vectors in phase $i$. Note that the cache configuration at each time slot is independent of the file request that arrives in that slot. Now, to upper bound the reward obtained by any policy $\pi$,
\begin{align*}
     \mathbb{E}\left( \sum_{t=1}^{T} \left \langle \boldsymbol{y}_{t}, \boldsymbol{x}_{t} \right \rangle \right) & = \sum_{i=1}^{T/r}\frac{1}{2C}\sum_{k=1}^{2C} y^{i}(k)\\
     &\leq \sum_{t=1}^{T/r}\frac{1}{2C}\sum_{k=1}^{L} y^{i}(k) \\
     &\leq \frac{T}{r}  \cdot \frac{1}{2C}\cdot rC \\&= \frac{T}{2},
\end{align*}
The second last step follows from the fact that the sum of the elements of $\boldsymbol{y}^i$ is exactly $rC$ as the cache configuration remains constant in each phase. 
\begin{align*}
    \therefore R_{A}^{\pi}(T) &\geq \mathbb{E}_{\{\boldsymbol{x}_{t}\}_{t=1}^{T}}\left(\boldsymbol{y}\text{*}\cdot \sum_{t=1}^{T}\boldsymbol{x}_{t}-\sum_{t=1}^{T}\boldsymbol{y}_{t}\cdot \boldsymbol{x}_{t}\right)\\
    &\geq \sqrt{\frac{CrT}{2\pi }}-\Theta\left({\frac{r\sqrt{r}}{\sqrt{T}}}\right).
\end{align*}
For this bound to be meaningful, the second term needs to be order-wise smaller than the first term: 
\begin{align*}
    \frac{r\sqrt{r}}{\sqrt{T}}<\sqrt{rT}
    \implies r<T
\end{align*}
Hence this is an $\mathcal{O}(\sqrt{rT})$ lower bound for $r<o(T)$. \\
For $r=\Omega(T)$, the first set of $r$ requests always gives a regret that is $\Omega(r)$ as in time slot 1, a random $C$ files out of $L$ are stored, and hence we get an overall regret of $\Omega(T)$. 

%
%
%

%
\end{document}